\documentclass[11pt]{llncs}

\usepackage{amssymb}
\usepackage{amsmath}
\usepackage{textcomp}
\usepackage{array}
\usepackage{lineno}
\usepackage{multirow}
\usepackage{verbatim}
\usepackage{color}
\usepackage{cite}

\newcommand{\bra}[1]{\langle #1|}
\newcommand{\ket}[1]{|#1\rangle}
\newcommand{\braket}[2]{\langle #1|#2\rangle}
\newcommand{\cent}[0]{\mbox{\textcent}}
\newcommand{\dollar}[0]{\$}

\newcommand{\leftstate}[1]{\overleftarrow{#1}}
\newcommand{\rightstate}[1]{\overrightarrow{#1}}
\newcommand{\stopstate}[1]{\downarrow\mspace{-2mu} #1}

\newcommand{\lc}[1]{|\overleftarrow{#1}\rangle}
\newcommand{\rc}[1]{|\overrightarrow{#1}\rangle}
\newcommand{\Sc}[1]{|\mspace{-5mu}\downarrow\mspace{-5mu} #1\rangle}

\newcommand{\lrhd}[0]{ \lhd \mspace{-3mu} \rhd }

\newtheorem{fact}{Fact}

\title{Unbounded-error quantum computation with small space bounds\thanks{A 
preliminary version of this paper appeared in the 
\textit{Proceedings of the Fourth International Computer Science Symposium in Russia}, 
pages 356--367, 2009.}$ ^{,} $\thanks{This work was
partially supported by the Bo\~{g}azi\c{c}i University Research Fund
with grant 08A102 and
by the Scientific and Technological Research Council of Turkey
(T\"{U}B\.ITAK) with grant 108E142.}}
\author{Abuzer Yakary{\i}lmaz\ \and A.C. Cem Say }
\institute{Bo\u{g}azi\c{c}i University, Department of Computer Engineering,\\ Bebek 34342 \.{I}stanbul, Turkey \\
\email{{abuzer,say}@boun.edu.tr}
 \\~~\\
\today
}

\begin{document}

\newlength{\twidth}
\maketitle
\pagenumbering{arabic}

\begin{abstract} \label{abstract:Abstract}

We prove the following facts about the language recognition power of
quantum Turing machines (QTMs) in the unbounded error setting: QTMs
are strictly more powerful than probabilistic Turing machines for any
common space bound $ s $ satisfying $ s(n)=o(\log \log n) $. For
``one-way" Turing machines, where the input tape head is not allowed
to move left, the above result holds for $s(n)=o(\log n) $. We also
give a characterization for the class of languages recognized with
unbounded error by real-time quantum finite automata (QFAs) with
restricted measurements. It turns out that these automata are equal in
power to their probabilistic counterparts, and this fact does not
change when the QFA model is augmented to allow general measurements
and mixed states. Unlike the case with classical finite automata, when
the QFA tape head is allowed to remain stationary in some steps, more languages
become recognizable. We define and use a QTM model that generalizes
the other variants introduced earlier in the study of quantum space
complexity.

\end{abstract}

\section{Introduction} \label{section:Introduction}

The investigation of the power of space-bounded quantum computers was
initiated by Watrous \cite{Wa98,Wa99,Wa03}, who defined several
machine models suitable for the analysis of this problem, and proved
that those quantum machines are equivalent to probabilistic Turing
machines (PTMs) for any common space-constructible bound $s(n) \in
\Omega(\log n) $ in the unbounded error case. Together with Kondacs,
Watrous also examined the case of constant space bounds, defining
\cite{KW97} a quantum finite automaton (QFA) variant, which inspired a
fruitful line of research \cite{Na99,AF98,AI99,Pa00,AW02,BP02,BMP03,ABGKMT06,Hi08,YS09B,YS10B}.

In this paper, we answer two open questions posed in the previous
study of space-bounded quantum complexity regarding sublogarithmic
space bounds. We first show that unbounded-error quantum Turing
machines are strictly more powerful than PTMs for any common space
bound $ s $ satisfying $s(n)=o(\log \log n) $. For ``one-way" Turing
machines, where the input
tape head is not allowed to move left, the above result holds for
$s(n)=o(\log n) $. We then give a full characterization of the class
of languages recognized with unbounded error by real-time QFAs with
restricted measurements. It turns out that these
automata have the same power as their classical counterparts, and this
fact does not change when the QFA definition is generalized in
accordance with the modern approach \cite{Hi08,AY10A}. Unlike the case with
classical finite automata, when the QFA tape head is allowed 
two-way movement, or even just the option of remaining stationary during some steps, more languages become recognizable.

As hinted above, early models of QTMs and QFAs \cite{Wa99,KW97,MC00}
were unduly restricted in their definitions, and did not reflect the
full potential of quantum mechanics in their computational power. This
problem was later addressed \cite{AW02,Wa03,Hi08} by the incorporation
of general quantum operations and mixed states into the models. Aiming
to provide the most general reasonable machine model for the study of
quantum space complexity, while taking into account the peculiarities
of small space bounds, we define a QTM variant of our own. The other
QTM models are shown to be specializations of our variant. We
conjecture that our model is more powerful than the other variants, at
least for some space bounds.

The rest of this paper is structured as follows: Section 2 contains
 relevant background information. The machine models we use,
including our new variant, are defined in Section 3. The superiority
of QTMs over PTMs for a range of small space bounds is established in
Section 4. In Section 5, we characterize the languages recognized with
unbounded error by all QFAs that are at least as powerful as the
``Kondacs-Watrous" model. Section 6 is a conclusion. Some technical details about our various quantum models are covered in the Appendices.

\section{Preliminaries} \label{section:preliminaries}
We start by introducing some notation and terminology that will be
used frequently in the remainder of the paper.
\subsection{Basic notation}

The following is a list of notational items that appear throughout the paper:
\begin{itemize}
       \item $ \Sigma $ is the input alphabet, not containing the end
markers $ \cent $ and $ \dollar $,
               and $ \tilde{\Sigma} = \Sigma \cup \{ \cent, \dollar \} $.
       \item $ \Gamma $ is the work tape alphabet, containing a
distinguished blank symbol denoted $ \# $.
       \item $ \Delta $ denotes the finite set of measurement outcomes.
       \item $ Q $ is the set of internal states,      where $ q_{1} $ is the
initial state.
       \item $ \delta $ is the transition function, which determines the
behaviour of the machine.
       \item $ \lrhd $ is the set $ \{ \leftarrow, \downarrow , \rightarrow
\} $, where
               $ \leftarrow $ means that the (corresponding) head moves one square
to the left,
               $ \downarrow $ means that the head stays on the same square, and
               $ \rightarrow $ means that the head moves one square to the right.
       \item $ f_{\mathcal{M}}(w) $ is the acceptance probability (or, in
one context, the acceptance value)
               associated by machine $ \mathcal{M} $ to input string $ w $ .
       \item $ \varepsilon $ is the empty string.
       \item For a given string $ w $, $ |w| $ is the length of $ w $,
               $ w_{i} $ is the $ i^{th} $ symbol of $ w $, and $ \tilde{w} $
represents the string $ \cent w \dollar $.
		\item $\mathbb{N} $ is the set of nonnegative integers.
		\item $  \mathbb{Z}^{+}$ is the set of positive integers.
       \item For a given (row or column) vector $ v $, $ v[i] $ is the $
i^{th} $ entry of $ v $.
       \item For a given matrix $ A $, $ A[i,j] $ is the $ (i,j)^{th} $
entry of $ A $.
       \item Some fundamental conventions in Hilbert space are as follows:
               \begin{itemize}
                       \item $ v $ and its conjugate transpose are denoted $ \ket{v} $
and $ \bra{v} $, respectively;
                       \item the multiplication of $ \bra{v_{1}} $  and $ \ket{v_{2}} $ is
shortly written as
                               $ \braket{v_{1}}{v_{2}} $;
                       \item the tensor product of $ \ket{v_{1}} $  and $ \ket{v_{2}} $
can also be written as
                               $ \ket{v_{1}} \ket{v_{2}} $ instead of $ \ket{v_{1}} \otimes \ket{v_{2}} $,
               \end{itemize}
               where $ v $, $ v_{1} $, and $ v_{2} $ are vectors.
\end{itemize}

\subsection{Language recognition}

The language $ L \subset \Sigma^{*} $ recognized by machine $
\mathcal{M} $ with (strict) cutpoint
$ \lambda \in \mathbb{R} $ is defined as
\begin{equation}
       L = \{ w \in \Sigma^{*} \mid f_{\mathcal{M}}(w) > \lambda \}.
\end{equation}
The language $ L \subset \Sigma^{*} $ recognized by machine $
\mathcal{M} $ with nonstrict cutpoint
$ \lambda \in \mathbb{R} $ is defined \cite{BJKP05} as
\begin{equation}
       L = \{ w \in \Sigma^{*} \mid f_{\mathcal{M}}(w) \geq \lambda \}.
\end{equation}
The language $ L \subset \Sigma^{*} $ is said to be recognized by
machine $ \mathcal{M} $ with unbounded error
if there exists a cutpoint  $ \lambda \in \mathbb{R} $ such that
$ L $ is recognized by $ \mathcal{M} $ with strict or nonstrict
cutpoint $ \lambda $.

\section{Space-bounded Turing machines} \label{section:space-bounded-TM}

The Turing machine (TM) models we use in this paper consist of a
read-only input tape with a two-way tape head, a read/write work tape
with a two-way tape head, and a finite state control. (The quantum
versions also have a finite register that plays a part in the
implementation of general quantum operations, and is used to determine
whether the computation has halted, and if so, with which decision.
For reasons of simplicity, this register is not included in the
definition of the probabilistic machines, since its functionality can
be emulated by a suitable partition of the set of internal states
without any loss of computational power.)
Both tapes are assumed to be two-way infinite and indexed by $ \mathbb{Z} $.

Let $ w $ be an input string.
On the input tape, $ \tilde{w} = \cent w_{1} \ldots w_{|w|} \dollar $
is placed in the squares indexed by $ 1, \ldots, | \tilde{w} | $,
and all remaining squares contain $ \# $.
When the computation starts, the internal state is $ q_{1} $,
and both heads are placed on the squares indexed by 1. Additionally, we assume that
the input tape head never visits the squares indexed by $ 0 $ or $
|\tilde{w}|+1 $.

The internal state and the symbols scanned on the input and work tapes
determine the transitions of the machine.
After each of these transitions,
the internal state is updated, the symbol on the work tape is overwritten, and
the positions of the input and work tape heads are updated with
respect to $ \lrhd $.
(In the quantum case, the content of the finite register is overwritten, too.)

A TM is said to be \textit{unidirectional} if the movements of input
and work tape heads are fixed
for each internal state to be entered in any transition.
That is, for a unidirectional TM, we have two functions
$ D_{i} : Q \rightarrow \lrhd $ and
$ D_{w} : Q \rightarrow \lrhd  $,
determining respectively the movements of the input and work tape heads.

A configuration of a TM is the collection of
\begin{itemize}
       \item the internal state of the machine,
       \item the position of the input tape head,
       \item the contents of the work tape, and the position of the work tape head.
\end{itemize}
$ \mathcal{C}^{w} $, or shortly $ \mathcal{C} $,
denotes the set of all configurations, which is a finite set in our
case of space bounded computations.
Let $ c_{i} $ and $ c_{j} $ be two configurations.
The probability (or amplitude) of the transition from $ c_{i} $ to $ c_{j} $ is 
given by the transition function $ \delta $ if
$ c_{i} $ is reachable from $ c_{j} $ in one step, and is zero otherwise.
(Note that, in probabilistic and quantum computation,
more than one outgoing transition can be defined for a single configuration.)
A \textit{configuration matrix} is a square matrix whose rows and columns are
indexed by the configurations.
The $ (j,i)^{th} $ entry of the matrix denotes the value of the
transition from $ c_{i} $ to $ c_{j} $.

We say that \cite{Wa98} a TM $ \mathcal{M} $ runs in space $ s $, for
$ s $ a function of the form
$ s: \mathbb{N} \rightarrow \mathbb{Z}^{+} $, if the following holds
for each input $ w $:
There exist $ s(|w|) $ contiguous tape squares on the work tape of $
\mathcal{M} $ such that
there is zero probability that the work tape head of $ \mathcal{M} $
leaves these tape squares
at any point in its computation on input $ w $.

By restricting the movement of the input tape head to the set $ \{ \downarrow, \rightarrow \} $, 
we obtain a one-way machine.

\subsection{Probabilistic Turing machines}
A probabilistic Turing machine (PTM) is a 7-tuple\footnote{Recall that some notation and terminology which will be used multiple times in this and the following definitions were introduced in Section 2.1.}
\begin{equation}
      \mathcal{P}=(Q,\Sigma,\Gamma,\delta,q_{1},Q_{a},Q_{r}),
\end{equation}
where $ Q_{a} $ and $ Q_{r} $, disjoint subsets of $ Q $ not including
$ q_{1} $, are the collections of accepting and rejecting
internal states, respectively.
Additionally, $ Q_{n} = Q \setminus \{ Q_{a} \cup Q_{r} \} $.

The transition function $ \delta $ is specified so that
\begin{equation}
       \delta(q,\sigma,\gamma,q^{\prime},\gamma^{\prime}) \in \tilde{\mathbb{R}}
\end{equation}
is the probability that the PTM will change its internal state to $
q^{\prime} $, write $ \gamma^{\prime} $ on the work tape,
and update the positions of the input and work tape heads with respect to
$ D_{i}(q^{\prime}) $ and $ D_{w}(q^{\prime}) $,
respectively,\footnote{We define PTMs as unidirectional machines. This
causes no loss of computational power, and increases the number of
internal states in our machines at most by a factor of 9.} if it
scans $ \sigma $ and $ \gamma $ on the input and work tapes,
respectively, when originally in internal state $ q $. $ \tilde{\mathbb{R} } $
is the set consisting of $ p \in \mathbb{R} $ such that there
is a deterministic
algorithm that computes $ p $ to within $2^{-n}$ in time polynomial in $n$.
We choose $ \tilde{\mathbb{R}} \cap [0,1] $ as our set of possible
transition probabilities, rather than the familiar
``coin-flipping" set $ \{0,\frac{1}{2}, 1 \} $, or the set of rational
numbers, since these possibilities are not known to be equivalent from
the point of view of computational power under the small space bounds
that we consider, and we wish to use the most powerful yet
``reasonable" models in our analysis.

For each input string $ w \in \Sigma^{*} $,
the transition function defines a unique configuration matrix, $ A^{w}
$, or shortly $ A $.
A PTM is \textit{well-formed} (i.e. fulfills the commonsense
requirement that the probabilities of alternative transitions always
add up to 1) if all columns of $ A $ are stochastic vectors.
This constraint defines the following \textit{local conditions for PTM
well-formedness}
that $ \delta $ must obey:
For each $ q \in Q $, $ \sigma \in \tilde{\Sigma} $, and $ \gamma \in \Gamma $,
\begin{equation}
      \sum_{q^{\prime},\gamma^{\prime}}
\delta(q,\sigma,\gamma,q^{\prime},\gamma^{\prime}) = 1,
\end{equation}
where $ q^{\prime} \in Q $ and $ \gamma^{\prime} \in \Gamma $.
A well-formed PTM can be described relatively easily by specifying
$\delta$ by presenting, for each $ \sigma \in \tilde{\Sigma} $,
a (left) stochastic transition matrix $ A_{\sigma} $,
whose rows and columns are indexed by (state, work tape symbol) pairs,
and the entry indexed by $ ((q^{\prime},\gamma^{\prime}),(q,\gamma)) $ equals
$ \delta(q,\sigma,\gamma,q^{\prime},\gamma^{\prime})  $.

The computation halts and the input is accepted (or rejected) whenever
the machine enters an internal state belonging to the set of
accepting (or rejecting) states.
PrSPACE$ (s) $ is the class of the languages that are recognized by a
PTM running in space $ O(s) $
with unbounded error.

The case of constant space bounds will be given special attention:
By removing the work tape of the PTM,\footnote{One only needs the
single ``direction" function $ D_{i} $ in this case.}
we obtain the two-way probabilistic finite state automaton (2PFA),
which is formally a 6-tuple
\begin{equation}
      \mathcal{P}=(Q,\Sigma,\delta,q_{1},Q_{a},Q_{r}).
\end{equation}
In this case, a well-formed machine can be specified by providing a (left) stochastic matrix
$ A_{\sigma} $, whose rows and columns are
indexed only by internal states, for each $ \sigma \in \tilde{\Sigma} $.

       In both probabilistic and quantum finite
       automata \cite{Ra63,Pa71,KW97}, the transition probabilities are traditionally allowed to be
       uncomputable numbers, and therefore the classes of recognized languages include
       undecidable ones \cite{Ra63}. TMs, however, are restricted to use
       computable transition probabilities, as seen in the definition above.
       Note that the simulation results in this paper do not change when we disallow
       finite automata to have uncomputable numbers as transition
       probabilities, since none of our constructions involve such numbers.

If we restrict the range of $ D_{i} $ in 2PFAs with $ \{ \rightarrow \} $,
we obtain the real-time probabilistic finite automaton (RT-PFA) model.
A RT-PFA can scan the input only once.
Traditionally, RT-PFAs are defined to be able to decide on acceptance or
rejection only
after the last symbol is read, and just specifying the set of
accepting states in their description is therefore sufficient, yielding a 5-tuple
\begin{equation}
      \mathcal{P}=(Q,\Sigma,\{ A_{\sigma} \mid \sigma \in \tilde{\Sigma} \},q_{1},Q_{a}).
\end{equation}
The computation of a RT-PFA can be traced by a stochastic state vector,
say $ v $, such that
$ v[i] $ corresponds to state $ q_{i} $.
For a given input string $ w \in \Sigma^{*} $,
($ \tilde{w} = \cent w \dollar $ is placed on the tape)
\begin{equation}
      v_{i} = A_{\tilde{w}_{i}} v_{i-1},
\end{equation}
where $ 1 \le i \le | \tilde{w} | $; $ \tilde{w}_{i} $ is the $ i^{th}
$ symbol of $ \tilde{w} $;
$ v_{0} $ is the initial state vector whose first entry is equal to 1.
The transition matrices of a RT-PFA can be extended for any string as 
\begin{equation}
      A_{w \sigma} = A_{\sigma} A_{w},
\end{equation}
where $ w \in \Sigma^{*} $, $ \sigma \in \tilde{\Sigma} $, and $ A_{\varepsilon} = I $.
The probability that $ w $ will be accepted by RT-PFA $ \mathcal{P} $ is
\begin{equation}
      f_{\mathcal{P}}(w) = \sum_{q_{i} \in Q_{a}} (A_{\tilde{w}}v_{0})[i] =
              \sum_{q_{i} \in Q_{a}} v_{|\tilde{w}|}[i].
\end{equation}

A generalization of the RT-PFA is the generalized finite automaton
(GFA), which is formally a 5-tuple
\begin{equation}
      \mathcal{G}=(Q,\Sigma,\{A_{\sigma} \mid \sigma \in \Sigma \},v_{0},f),
\end{equation}
where
\begin{enumerate}
      \item $ A_{\sigma} $'s are $ |Q| \times |Q| $-dimensional
real valued transition matrices;
      \item $ v_{0} $ and $ f $ are real valued  \textit{initial} (column)
and \textit{final} (row) vectors,
              respectively.
\end{enumerate}
Similar to what we had for RT-PFAs, the transition matrices of a GFA can
be extended for any string.
For a given input string, $ w \in \Sigma^{*} $, the acceptance value
associated by GFA
$ \mathcal{G} $ to string $ w $ is
\begin{equation}
      f_{\mathcal{G}}(w)=f A_{w_{|w|}} \cdots A_{w_{1}} v_{0} = f A_{w} v_{0}.
\end{equation}

RT-PFAs, GFAs \cite{Tu69}, and 2PFAs \cite{Ka91} recognize the same
class of languages with cutpoint. This is the class of
\textit{stochastic languages}, denoted by S. The class of languages
recognized by these machines with nonstrict cutpoint is denoted by
coS. The class of languages recognized by
RT-PFAs, GFAs, and 2PFAs with unbounded error is therefore
S $ \cup $ coS, and is denoted by uS. Note that PrSPACE(1) $
\subsetneq $ uS,
since uS contains undecidable languages.

\subsection{Quantum Turing machines}

We define a quantum Turing machine (QTM) $ \mathcal{M} $ to be a 7-tuple
\begin{equation}
      M=(Q,\Sigma,\Gamma,\Omega,\delta,q_{1},\Delta),
\end{equation}
which is distinguished from the PTM by the presence of the items $\Omega$, the finite register alphabet, containing the special
initial symbol $ \omega_{1} $, and $ \Delta = \{ \tau_{1}, \ldots , \tau_{k} \} $, the set of possible outcomes associated with the measurements of
the finite register. $ \Omega $ is partitioned into $ |\Delta| = k $ subsets 
$ \Omega_{\tau_{1}}, \ldots , \Omega_{\tau_{k}} $.

In accordance with quantum theory, a QTM can be in a superposition of
more than one configuration at the same time. The ``weight"
of each configuration in such a superposition is called its amplitude.
Unlike the case with PTMs, these amplitudes are not restricted to
being positive real numbers, and that is what gives quantum computers
their interesting features. A superposition of configurations
\begin{equation}
       \ket{\psi} = \alpha_{1} \ket{c_{1}} + \alpha_{2} \ket{c_{2}} + \cdots
+  \alpha_{n} \ket{c_{n}}
\end{equation}
can be represented by
a column vector $ \ket{\psi} $ with a row for each possible configuration,
where the $ i^{th} $ row contains the amplitude of the corresponding
configuration in $ \ket{\psi} $.

If our knowledge that the quantum system under consideration is in
superposition $ \ket{\psi} $ is certain, then $ \ket{\psi} $ is called
a \textit{pure
state}, and the vector notation described above is a suitable way of
manipulating this information. However, in some cases (e.g. during
classical probabilistic computation), we only know that the system is
in state $ \ket{\psi_{l}} $ with probability $ p_{l} $ for an ensemble of pure
states $ \{ (p_{l},\ket{\psi_{l}}) \} $,where
$ \sum_{l} p_{l}=1 $. A convenient representation tool for describing
quantum systems in such \textit{mixed states}
is the density matrix.
The \textit{density matrix}\footnote{Density matrices are Hermitian positive semidefinite matrices of trace 1.} representation of
$ \{ (p_{l},\ket{\psi_{l}}) \mid 1 \le l \le M < \infty \} $ is
\begin{equation}
      \rho = \sum_{l} p_{l} \ket{\psi_{l}} \bra{\psi_{l}}.
\end{equation}
We will use both these representations for quantum states in this
paper. We refer the reader to \cite{NC00} for further details.

The initial density matrix of the QTM is represented by  $ \rho_{0} =
\ket{c_{1}} \bra{c_{1}} $,
where $ c_{1} $ is the initial configuration corresponding to the
given input string.

The transition function of a QTM is specified so that
\begin{equation}
      \delta(q,\sigma,\gamma,q^{\prime},d_{i},\gamma^{\prime},d_{w},\omega)
\in \tilde{\mathbb{C}}
\end{equation}
is the amplitude with which the QTM will change its internal state to
$ q^{\prime} $, write $ \gamma^{\prime} $ on the work tape
and $ \omega $ in the finite register, and update the positions of the
input and work tape heads with respect to $ d_{i} $ and $ d_{w} $,
respectively,
where $ d_{i},d_{w} \in \lrhd $, if it scans $ \sigma $ and $ \gamma $ on the
input and work tapes, respectively, when originally in internal state
$ q $. $ \tilde{\mathbb{C}} $ \cite{BV97} is the set of complex
numbers whose real and imaginary parts are in
$ \tilde{\mathbb{R}} $.

After each transition, the finite register is measured \cite{NC00} as described by
the set of operators
\begin{equation}
      P=\left\lbrace P_{\tau} \mid P_{\tau} =
              \sum_{\omega \in \Omega_{\tau}}
              \ket{\omega}\bra{\omega}, \tau \in \Delta \right\rbrace.
\end{equation}
In its standard usage, $ \Delta $ is the set $ \{a,n,r\} $, and the
following actions are associated with the measurement outcomes:
\begin{itemize}
      \item ``$n$": the computation continues;
      \item ``$a$": the computation halts, and the input is accepted;
      \item ``$r$": the computation halts, and the input is rejected.
\end{itemize}
The finite register is reinitialized to $ \omega_{1}$, irreversibly
erasing its previous content, before the next transition of the
machine.

Since we do not consider the register content as part of the
configuration, the register can be seen as the
``environment" interacting with the ``principal system" that is the
rest of the QTM \cite{NC00}.
The transition function $ \delta $ therefore induces a set of configuration transition
matrices, $ \{ E_{\omega} \mid \omega \in \Omega \} $, where the $ (i,j)^{th} $ entry of $ E_{\omega}$,
the amplitude of the transition  from $ c_{j} $ to $ c_{i} $ by
writing $ \omega \in \Omega $ on the register,
is defined by $ \delta $ whenever $ c_{j} $ is reachable from $ c_{i}
$  in one step, and is zero otherwise. The $ \{ E_{\omega} \mid \omega \in \Omega \} $ form an operator $ \mathcal{E} $, with operation elements
$  \mathcal{E}_{\tau_{1}} \cup \mathcal{E}_{\tau_{2}} \cup \cdots \cup
\mathcal{E}_{\tau_{k}}  $, where for each $ \tau \in \Delta $,
$ \mathcal{E}_{\tau} = \{ E_{\omega} \mid \omega \in \Omega_{\tau} \} $.

According to the modern understanding of quantum computation
\cite{AKN98}, a QTM is said to be
\textit{well-formed}\footnote{We also refer the reader to \cite{BV97}
for a detailed discussion
of the well-formedness of QTMs that evolve unitarily.}
if $ \mathcal{E} $ is a superoperator (selective quantum operator), i.e.
\begin{equation}
      \sum_{\omega \in \Omega} E_{\omega}^{\dagger}E_{\omega} = I.
\end{equation}
$ \mathcal{E} $ can be represented by a $ | \mathcal{C} | |
\Omega | \times | \mathcal{C} | $-dimensional
matrix $ \mathsf{E} $ (Figure \ref{figure:matrix-E}) by concatenating each $ E_{\omega} $ one under the other,
where $ \omega \in \Omega $. 
It can be verified that $ \mathcal{E} $ is a superoperator if and only
if the columns of $ \mathsf{E} $
form an orthonormal set.

\begin{center}
\begin{figure}[h]
	\centering	
	\begin{minipage}{0.3\textwidth}
		\[
		\begin{array}{ccccc}
			& c_{1} & c_{2} & \ldots & c_{|\mathcal{C}|} \\
			\hline
			\multicolumn{1}{c|}{c_{1}} & & & & \multicolumn{1}{c|}{} \\
			\multicolumn{1}{c|}{c_{2}} & & & & \multicolumn{1}{c|}{} \\
			\multicolumn{1}{c|}{\vdots} & \multicolumn{4}{c|}{ E_{\omega_{1}} } \\
			\multicolumn{1}{c|}{c_{|\mathcal{C}|}} & & & & \multicolumn{1}{c|}{} \\
			\hline	
			\multicolumn{1}{c|}{c_{1}} & & & & \multicolumn{1}{c|}{} \\
			\multicolumn{1}{c|}{c_{2}} & & & & \multicolumn{1}{c|}{} \\
			\multicolumn{1}{c|}{\vdots} & \multicolumn{4}{c|}{ E_{\omega_{2}} } \\
			\multicolumn{1}{c|}{c_{|\mathcal{C}|}} & & & & \multicolumn{1}{c|}{} \\
			\hline\multicolumn{1}{c|}{c_{1}} & & & & \multicolumn{1}{c|}{} \\
			\multicolumn{1}{c|}{c_{2}} & & & & \multicolumn{1}{c|}{} \\
			\multicolumn{1}{c|}{\vdots} & \multicolumn{4}{c|}{\vdots } \\
			\multicolumn{1}{c|}{c_{|\mathcal{C}|}} & & & & \multicolumn{1}{c|}{} \\
			\hline	
			\multicolumn{1}{c|}{c_{1}} & & & & \multicolumn{1}{c|}{} \\
			\multicolumn{1}{c|}{c_{2}} & & & & \multicolumn{1}{c|}{} \\
			\multicolumn{1}{c|}{\vdots} & \multicolumn{4}{c|}{ E_{\omega_{|\Omega|}} } \\
			\multicolumn{1}{c|}{c_{|\mathcal{C}|}} & & & & \multicolumn{1}{c|}{} \\
			\hline	
		\end{array}
		\]
	\end{minipage}
	\caption{Matrix $ \mathsf{E} $}
	\label{figure:matrix-E}
\end{figure}
\end{center}

Let $ c_{j_{1}} $ and $ c_{j_{2}} $ be two configurations with corresponding columns 
$ v_{j_{1}} $ and $ v_{j_{2}} $ in $ \mathsf{E} $.
For an orthonormal set to be formed, we must have
\begin{equation}
      v_{j_{1}}^{\dagger} \cdot  v_{j_{2}} = \left \lbrace \begin{array}{ll}
              1 & j_{1}=j_{2} \\
              0 & j_{1} \neq j_{2}
      \end{array}
      \right.
\end{equation}
for all such pairs.
This constraint imposes some easily checkable restrictions on $ \delta $.
The (quite long) list of these local conditions for QTM
wellformedness can be found in \cite{Ya11A}.

PrQSPACE$ (s) $ is the class of languages that are recognized by QTMs running in space $ O(s) $
with unbounded error. (Note that this complexity class has been
defined and used by Watrous in references \cite{Wa99,Wa03}. As we
will demonstrate shortly, our QTM model is at least as powerful as the
models used in those papers, and it may well be strictly more powerful
than them. Since the aim is to understand the full power of
space-bounded quantum computation, we suggest it would make sense to
adopt our definition of PrQSPACE as the standard.)

It is a well-established fact \cite{Wa09} that 
any quantum computational model defined using superoperators can efficiently simulate its classical 
counterpart, and so PrSPACE($s$) $ \subseteq $ PrQSPACE($s$) for all $s$. Some early models 
used in the study of space-bounded quantum computation, which do not make full use of the 
capabilities allowed by quantum mechanics, can fail to achieve some tasks that are possible for 
the corresponding classical machines \cite{KW97,MC00}.

The two-way quantum finite automaton (2QFA) is obtained by removing the
work tape of the QTM:
\begin{equation}
      \mathcal{M}=(Q,\Sigma,\Omega,\delta,q_{1},\Delta).
\end{equation}
The transition function of a 2QFA is therefore specified so that
\begin{equation}
      \delta(q,\sigma,q^{\prime},d_{i},\omega) \in \mathbb{C}
\end{equation}
is the amplitude with which the machine enters
state $ q^{\prime} $, writes $ \omega $ on the register,
and updates the position of the input tape with respect to $ d_{i} \in
\lrhd $, if it reads $ \sigma $ on the input
tape when originally in state $ q $. 
See Appendix \ref{appendix:QTM-wellformedness} for a list of easily checkable 
local conditions for wellformedness of 2QFAs.

In the remainder of this section, we
will examine some specializations of the QTM model that have appeared
in the literature.

\subsubsection*{QTMs with classical heads}

Although our definition of space usage as the number of work tape
squares used during the computation is standard in the study of small
as well as large space bounds \cite{Si06,Sz94,FK94},
some researchers prefer to utilize QTM models where the tape head
locations are classical (i.e. the heads do not enter quantum
superpositions) to avoid the possibility of using quantum resources
that increase with input size for the implementation of the heads.
For details of this specialization of our model, which we call the QTM with classical heads (CQTM), see Appendix B, which also includes a demonstration
of the fact that all quantum machines can simulate their probabilistic
counterparts easily.

Watrous' QTM model in \cite{Wa03}, which we call Wa03-QTM for ease of
reference, is a CQTM variant
that has an additional classical work tape and classical internal
states. Every Wa03-QTM can be simulated exactly (i.e. preserving the
same acceptance probability for every input) by CQTMs with only some
time overhead.\footnote{We
omit the proof here, but it is not hard to show how to simulate the
classical components of a Wa03-QTM within a CQTM.} Note that
Wa03-QTMs allow only algebraic transition amplitudes by definition.

Let us consider real-time versions of 2QFAs, 
whose tape heads are forced by definition to have classical locations \cite{Pa00}. 
 If the quantum machine model used is sufficiently general, 
then the intermediate measurements can be postponed easily to the end of
the algorithm in real-time computation. That final measurement can be performed on the set of
internal states, rather than the finite register.
Therefore, as with RT-PFAs, we specify a subset of the internal
states of the machine as the collection of accepting states, denoted $ Q_{a} $.

A real-time quantum finite automaton (RT-QFA) \cite{Hi08} is a 5-tuple
\begin{equation}
      \mathcal{M}=(Q,\Sigma,\{\mathcal{E}_{\sigma} \mid \sigma \in \tilde{\Sigma} \},q_{1},Q_{a}),
\end{equation}
where each $ \mathcal{E}_{\sigma } $ is an operator having elements
$ \{ E_{\sigma,1},\ldots,E_{\sigma,k} \} $ for some $ k \in \mathbb{Z}^{+} $
satisfying
\begin{equation}
      \sum_{i=1}^{k} E_{\sigma,i}^{\dagger} E_{\sigma,i} = I.
\end{equation}
Additionally, we define the projector
\begin{equation}
      P_{a} = \sum_{q \in Q_{a}} \ket{q}\bra{q}
\end{equation}
in order to check for acceptance.
For any given input string $ w \in \Sigma^{*} $, $ \tilde{w} $ is
placed on the tape, and the computation can be traced by density matrices
\begin{equation}
 \label{equation:denmat}
      \rho_{j} = \mathcal{E}_{\tilde{w}_{j}} (\rho_{j-1}) =
      \sum_{i=1}^{k} E_{\tilde{w}_{j},i} \rho_{j-1} E_{\tilde{w}_{j},i}^{\dagger},
\end{equation}
where $ 1 \le j \le | \tilde{w} |  $, and $ \rho_{0} = \ket{q_{1}}
\bra{q_{1}} $ is the initial density matrix. This is how density
matrices evolve according to superoperators \cite{NC00}.
The transition operators can be extended easily for any string as
\begin{equation}
      \mathcal{E}_{w \sigma} = \mathcal{E}_{\sigma} \circ \mathcal{E}_{w},
\end{equation}
where $ w \in \Sigma^{*} $ and $ \mathcal{E}_{\varepsilon} = I $.
Note that, if $ \mathcal{E}=\{E_{i} \mid 1 \le i \le k \} $ and 
$ \mathcal{E}^{\prime}=\{ E^{\prime}_{j} \mid 1 \le j \le k^{\prime} \} $, then
\begin{equation}
	\mathcal{E^{\prime}} \circ \mathcal{E} = \{ E_{j}^{\prime}E_{i} \mid 1 \le i \le k, 1 \le j \le k^{\prime} \}.
\end{equation}
The probability that RT-QFA $ \mathcal{M} $ will accept $ w $ is
\begin{equation}
      f_{\mathcal{M}}(w) = tr( P_{a} \mathcal{E}_{\tilde{w}}(\rho_{0})) = tr(P_{a} \rho_{| \tilde{w} |} ).
\end{equation}

The class of languages recognized by RT-QFAs with cutpoint is denoted by QAL. 
The class of languages recognized by these machines with nonstrict cutpoint is denoted by coQAL.
QAL $ \cup $ coQAL is denoted by uQAL.

\begin{lemma}
      \label{lemma:RT-QFA-to-GFA}
      For any RT-QFA $ \mathcal{M} $ with $ n $ internal states,
      there exists a GFA  $ \mathcal{G} $ with $ n^{2} $ internal
states such that
      $ f_{\mathcal{M}}(w) = f_{\mathcal{G}}(w) $ for all $ w \in \Sigma^{*} $.
\end{lemma}
\begin{proof}
      Let $ \mathcal{M} = (Q_{1},\Sigma,\{\mathcal{E}_{\sigma} \mid \sigma \in \tilde{\Sigma}\},q_{1},Q_{a} ) $ 
      be the RT-QFA with $ n $ internal states, and let each 
      $ \mathcal{E}_{\sigma} $ have $ k $ elements, without loss of generality.
      We will construct GFA $ \mathcal{G}=(Q_{2},\Sigma,\{A_{\sigma} \mid \sigma \in \Sigma \},v_{0},f) $ 
      with $ n^{2} $ internal states. We start by building an intermediate GFA $ \mathcal{G^{\prime}}=(Q_{3},\Sigma,\{A^{\prime}_{\sigma} \mid \sigma \in \Sigma \},v^{\prime}_{0},f^{\prime}) $ with the required simulation property but
      with $ 2n^{2} $ states.
      We will use the mapping $ vec $ described in Figure \ref{figure:vec}
      in order to linearize the computation of $ \mathcal{M} $,
      so that it can be traced by  $ \mathcal{G}^{\prime} $.

\begin{figure}[h]      
      \centering
      \fbox{
      \begin{minipage}{0.90\textwidth}
              \footnotesize{
                      Let $ A $, $ B $, and $ C $      be $ n \times
n $ dimensional matrices.
                      $ vec $ is a linear mapping from $ n \times n $
                      matrices to $ n^{2} $ dimensional (column)
vectors defined as
                      \begin{equation}
                              vec(A)[(i-1)n+j] = A[i,j],
                      \end{equation}
                      where $ 1 \le i,j \le N $. One can
verify the
following properties:
                      \begin{equation}
                              \label{equation:vec-ABC}
                              vec(ABC) = (A \otimes C^{T})vec(B)
                      \end{equation}
                      and
                      \begin{equation}
                              \label{equation:vec-AtB}
                              tr(A^{T}B)=vec(A)^{T}vec(B).
                      \end{equation}
              }
      \end{minipage}
      }
      \caption{The definition and properties of $ vec $ (see Page 73 in \cite{Wa03})}
      \label{figure:vec}
\end{figure}

      We define
      \begin{equation}
              v_{0}^{\prime\prime} = vec(\rho_{1}),
      \end{equation}
      where
      \begin{equation}
              \rho_{1} = \mathcal{E}_{\cent}(\rho_{0}) = \sum_{i=1}^{k} E_{\cent,i} \rho_{0} E_{\cent,i}^{\dagger}.
      \end{equation}
      For each $ \sigma \in \Sigma $, we define
      \begin{equation}
              A_{\sigma}^{\prime\prime} = \sum_{i=1}^{k} E_{\sigma,i} \otimes E_{\sigma,i}^{*},
      \end{equation}
      and so we obtain (by Equations \ref{equation:denmat} and \ref{equation:vec-ABC})
      \begin{equation}
              vec(\mathcal{E_{\sigma}}(\rho)) = A_{\sigma}^{\prime\prime} vec(\rho)
      \end{equation}
      for any density matrix $ \rho $.
      Finally, we define
      \begin{equation}
              f^{\prime\prime} = vec(P_{a})^{T} \sum_{i=1}^{k} E_{\dollar,i} \otimes E_{\dollar,i}^{*}.
      \end{equation}
      It can be verified by using Equation \ref{equation:vec-AtB} that for any input string $ w \in \Sigma^{*} $,
      \begin{equation}
              f^{\prime\prime} A^{\prime\prime}_{w_{|w|}} \cdots A^{\prime\prime}_{w_{1}} v_{0}^{\prime\prime} =
              tr(P_{a} \mathcal{E}_{\dollar} \circ \mathcal{E}_{w} \circ \mathcal{E}_{\cent} (\rho_{0})) =
              f_{\mathcal{M}}(w).
      \end{equation}
      The complex entries of $ v_{0}^{\prime\prime} $, 
      $ \{ A^{\prime\prime}_{\sigma} \mid \sigma \in \Sigma \} $, and $ f^{\prime\prime} $
      can be replaced \cite{MC00} with $ 2 \times 2 $ dimensional real matrices,\footnote{$ a+bi $ is replaced with
      $ \left( \begin{array}{rl}      a & b \\ -b & a \end{array} \right) $.} and so we obtain the equations
      \begin{equation}
              \left(
              \begin{array}{cc}
                      f_{\mathcal{M}}(w) & 0 \\
                      0 & f_{\mathcal{M}}(w)
              \end{array}
              \right) = f^{\prime\prime\prime}  A^{\prime\prime\prime}_{w_{|w|}} \cdots A^{\prime\prime\prime}_{w_{1}} 
              v_{0}^{\prime\prime\prime},
      \end{equation}
      where the terms with triple primes are obtained from the corresponding terms with double primes.
      We finish the construction of $ \mathcal{G^{\prime}} $ by stating that
      \begin{enumerate}
              \item $ v_{0}^{\prime} $ is the first column of $ v_{0}^{\prime\prime\prime} $,
              \item $ A_{\sigma}^{\prime} $ is equal to $ A_{\sigma}^{\prime\prime\prime} $, for each $ \sigma \in \Sigma $, and
              \item $ f^{\prime} $ is the first row of $ f^{\prime\prime\prime} $.
      \end{enumerate}
    We  refer the reader to \cite{MC00,LQ08,YS09A}, that present similar
	constructions for other
	types of real-time QFAs. The remainder of this proof is an improvement over these constructions regarding the number of states, and was kindly suggested to us by one of the anonymous referees of this paper, to whom we are indebted.
	
	Since density matrices are Hermitian, all entries on the main diagonal are real, and the entries on the opposite sides of the diagonal are complex conjugates of each other, meaning that one actually needs only $n^2$ distinct real numbers to represent the $ n \times n $ density matrices of $\mathcal{M}$. So the information in the vector $ v_{0}^{\prime} $ in the definition of $ \mathcal{G^{\prime}} $ can in fact fit in an $ n^{2} $-dimensional vector.
	To perform conversions between these two representations, we can define two linear operators,
	denoted $ L $ and $ L' $, such that
	\begin{itemize}
		\item $ L $, an $ n^{2} \times 2n^{2} $-dimensional matrix containing entries from the set $\{-1, 0, 1\}$,
			transforms $ 2n^{2} $-dimensional vectors in the format of machine $ \mathcal{G^{\prime}} $ to equivalent $ n^{2} $-dimensional vectors, and
		\item $ L' $, a $ 2n^{2} \times n^{2} $-dimensional matrix, 
			performs the reverse transformation.
	\end{itemize}
Hence, the state-efficient GFA	$ \mathcal{G} $ is constructed by setting
      \begin{enumerate}
              \item $ v_{0} = L v_{0}^{\prime} $,
              \item $ A_{\sigma} = L A_{\sigma}^{\prime} L' $, for each $ \sigma \in \Sigma $, and
              \item $ f = f^{\prime} L' $.
      \end{enumerate}
\qed\end{proof}

\begin{corollary}
      QAL = S.
\end{corollary}
We therefore have that real-time unbounded-error probabilistic and quantum finite automata
are equivalent in power.
We will show in Sections \ref{section:MainResults} and \ref{section:KWQFA-languages} 
that this equivalence does not carry over to the two-way case.

\subsubsection*{QTMs with restricted measurements}

In another specialization of the QTM model, the \textit{QTM with
restricted measurements}, the machine is unidirectional, the
heads can enter quantum superpositions, $ \Delta = \{n,a,r\} $, and
$ | \Omega_{n} | = | \Omega_{a} | = | \Omega_{r} | = 1  $. The first
family of QTMs that was formulated for the analysis of space
complexity issues \cite{Wa98,Wa99}, which we call the Wa98-QTM,
corresponds to such a model, with the added restriction that the
transition
amplitudes are only allowed to be rational numbers.
The finite automaton versions of QTMs with restricted
measurements\footnote{These models,
which also allow unrestricted transition amplitudes by the convention
in automata theory, are introduced in the paper written by Kondacs and
Watrous \cite{KW97}.}
are known as Kondacs-Watrous quantum finite automata, and abbreviated
as 2KWQFAs, 1KWQFAs, or RT-KWQFAs, depending on the set of allowed
directions of movement for the input head.
These are pure state models,
since the non-halting part of the computation is always represented by
a single quantum state.
Therefore, configuration or state vectors, rather than the density
matrix formalism, can be used in order to
trace the computation easily.
To be consistent with the literature on 2KWQFAs, we specialize the
2QFA model by the following process:
\begin{enumerate}
	\item The finite register does not need to be refreshed, since the 
       	computation continues if and only if the initial symbol is observed.
    \item In fact, 2KWQFAs do not need to have the finite register at all, instead, similarly to 2PFAs, the set of internal states of the 2KWQFA 
    is partitioned to sets of nonhalting, accepting, and rejecting states, denoted
    $ Q_{n} $, $ Q_{a} $, and $ Q_{r} $, respectively, which can be obtained easily by taking the tensor
    product of the internal states of the 2QFA and the set $ \{n, a, r\} $.
	\item A configuration is designated as nonhalting (resp. accepting or rejecting),
		if its internal state is a member of $ Q_{n} $ (resp. $ Q_{a} $ or $ Q_{r} $).
               Nonhalting (resp. accepting or rejecting) configurations form the set $ \mathcal{C}_{n} $
               (resp. $ \mathcal{C}_{a} $ or $ \mathcal{C}_{r} $) (for a given input string).
       \item The evolution of the configuration sets can be represented by a
unitary matrix.
       \item The measurement is done on the configuration set with
projectors $ P_{n} $, $ P_{a} $, and $ P_{r}$,
              defined as
              \begin{equation}
                      P_{\tau}  =\sum_{c \in \mathcal{C}^{w}_{\tau}} \ket{c}\bra{c}
              \end{equation}
              for a given input string $ w \in \Sigma^{*} $, where $ \tau \in \{n,a,r\} $ and 
              the standard actions are associated with the outcomes ``$n$", ``$a$", and ``$r$".
\end{enumerate}

Formally, a 2KWQFA is a 6-tuple
\begin{equation}
       \mathcal{M} = \{Q,\Sigma,\delta,q_{1},Q_{a},Q_{r}\},
\end{equation}
where $ Q_{n} = Q \setminus \{ Q_{a} \cup Q_{r} \} $ and $ q_{1} \in Q_{n} $.
$ \delta $ induces a unitary matrix $ U_{\sigma}$, whose rows and columns are
indexed by internal states for each input symbol $ \sigma $.
Since all 2KWQFAs are unidirectional, 
we will use the notations $ \leftstate{q} $, $ \stopstate{q} $, and $
\rightstate{q} $
for internal state $ q $ in order to represent the value of $ D_{i}(q) $ as
$ \leftarrow $, $ \downarrow $, and $ \rightarrow $, respectively.

A RT-KWQFA is a 6-tuple
\begin{equation}
       \mathcal{M} = \{Q,\Sigma,\{U_{\sigma} \mid \sigma \in \tilde{\Sigma} \},q_{1},Q_{a},Q_{r}\},
\end{equation}
where $ \{U_{\sigma} \mid \sigma \in \tilde{\Sigma} \} $ are unitary transition matrices.
In contrast to the other kinds of real-time finite automata, a RT-KWQFA is
measured at each step during computation
after the unitary transformation is applied.
The projectors are defined as
\begin{equation}
       P_{\tau}     =\sum_{q \in Q_{\tau}} \ket{q}\bra{q},
\end{equation}
where $ \tau \in \Delta $.
The nonhalting portion of the computation of a RT-KWQFA can be traced by
a state vector, say $ \ket{u} $, such that
$ \braket{i}{u} $ corresponds to state $ q_{i} $.
The computation begins with $ \ket{u_{0}} = \ket{q_{1}} $. For a given
input string $ w \in \Sigma^{*} $,
at step $ j $ $ (1 \le j \le |\tilde{w}|) $:
\begin{equation}
       \ket{u_{j}} = P_{n}U_{\tilde{w}_{j}} \ket{u_{j-1}},
\end{equation}
the input is accepted with probability
\begin{equation}
       || P_{a}U_{\tilde{w}_{j}} \ket{u_{j-1}} ||^{2},
\end{equation}
and rejected with probability
\begin{equation}
       || P_{r}U_{\tilde{w}_{j}} \ket{u_{j-1}} ||^{2}.
\end{equation}
The overall acceptance and rejection probabilities are accumulated by
summing up these values at each step. Note that, the state vector
representing the nonhalting portion is not normalized in the
description given above.

Brodsky and Pippenger \cite{BP02}, who studied various properties of
some early models of quantum finite automata, defined the class of
languages recognized by RT-KWQFAs with unbounded error, denoted $ UMM
$,
in a way that is slightly different than our approach in this paper:
$ L \in UMM $ if and only if there exists a RT-KWQFA $ \mathcal{M} $
such that
\begin{itemize}
       \item $ f_{\mathcal{M}}(w) > \lambda $ when $ w \in L $ and
       \item $ f_{\mathcal{M}}(w) < \lambda $ when $ w \notin L $,
\end{itemize}
for some $ \lambda \in [0,1] $.

For descriptions of several other QTM variants, we refer the reader to \cite{MW08} and \cite{Ga09}.

\section{Probabilistic vs. quantum computation with sublogarithmic space} \label{section:MainResults}

Watrous compared the unbounded-error probabilistic space complexity
classes (PrSPACE($s$)) with the corresponding classes for both Wa98-QTMs \cite{Wa98,Wa99} and 
Wa03-QTMs \cite{Wa03} for space bounds $ s=\Omega(\log n) $, 
establishing the identity of the associated
quantum space complexity classes with each other, and also with the
corresponding probabilistic ones. The case of $s=o(\log n)$ was left
as an open question \cite{Wa99}. 
In this section, we provide an answer to that question.

We already know that QTMs allowing superoperators are at least as powerful as PTMs for any common space bound. 
We will now exhibit a 1KWQFA which performs a task that is impossible
for PTMs with small space bounds.

Consider the nonstochastic context-free language \cite{NH71}	
\begin{equation*}
	\footnotesize
	\mbox{ $
	L_{NH} = \{a^{x}ba^{y_{1}}ba^{y_{2}}b \cdots a^{y_{t}}b \mid x,t,y_{1}, \cdots, y_{t}
		\in \mathbb{Z}^{+} \mbox{ and } \exists k ~ (1 \le k \le t), x=\sum_{i=1}^{k}y_{i} \}
		$ }
\end{equation*}
over the alphabet $ \{a,b \} $.
Freivalds and Karpinski \cite{FK94} have
proven the following facts about $ L_{NH} $:
\begin{fact}
	No PTM using space $ o(\log\log n) $ can recognize $ L_{NH} $ with unbounded error.
\end{fact}
\begin{fact}
	No 1PTM using space $ o(\log n) $ can recognize $ L_{NH} $ with unbounded error.
\end{fact}

There exists a one-way \textit{deterministic} TM that recognizes $
L_{NH} $ within the optimal space bound $O(\log n)$ \cite{FK94}. No
(two-way) PTM which recognizes $ L_{NH} $ using $ o(\log n) $ space is
known as of the time of writing.

\begin{theorem}
       \label{theorem:1KWQFA}
       There exists a 1KWQFA that recognizes $ L_{NH} $ with unbounded error.
\end{theorem}
\begin{proof}
       Consider the 1KWQFA $ \mathcal{M}=(Q, \Sigma, \delta, q_{0},Q_{a},Q_{r}) $, 
       where $ \Sigma=\{a,b\} $, and the state sets are as follows:
       \[
               \begin{array}{lcl}
                       Q_{n} & = &
                               \{ \rightstate{q_{0}} \} \cup \{ \rightstate{q_{i}} \mid 1 \le i \le 6 \} \cup \{
\rightstate{p_{i}} \mid 1 \le i \le 6 \}
                               \cup \{ \rightstate{a_{i}} \mid 1 \le i \le 4 \} \\
                               & & \cup~\{ \rightstate{r_{i}} \mid 1 \le i \le 4 \}
                                       \cup \{ \stopstate{w_{i}} \mid 1 \le i \le 6 \}, \\
                       Q_{a} &  =  & \{ \stopstate{A_{i}} \mid 1 \le i \le 18  \},
                               ~Q_{r} = \{ \stopstate{R_{i}} \mid 1 \le i \le 18  \}. \\
               \end{array}
        \]
       Let each $ U_{\sigma} $ induced by $ \delta $ act as indicated in Figures
\ref{figure:1KWQFA-1} and \ref{figure:1KWQFA-2}, and extend each to be unitary.
\begin{figure}[!h]       
       \setlength{\extrarowheight}{1pt}
       \centering
\footnotesize{
\begin{tabular}{|c|l|l|}
\hline
       Stages & \multicolumn{1}{c|}{$ U_{\cent}, U_{a} $} &
                \multicolumn{1}{c|}{$ U_{\dollar} $} \\
\hline
       & $ U_{\cent} \rc{q_{0}} = \frac{1}{\sqrt{2}}\rc{q_{1}} +
\frac{1}{\sqrt{2}}\rc{p_{1}} $  & \\
\hline
       $ \begin{array}{@{}c@{}} \mbox{I} \\ ( \mathsf{path_{1}} ) \end{array} $
       &
       $ \begin{array}{@{}l@{}}
               U_{a} \rc{q_{1}} = \frac{1}{\sqrt{2}}\rc{q_{2}} + \frac{1}{2}
\Sc{A_{1}} + \frac{1}{2} \Sc{R_{1}} \\
               U_{a} \rc{q_{2}} = \frac{1}{\sqrt{2}}\rc{q_{2}} - \frac{1}{2}
\Sc{A_{1}} - \frac{1}{2} \Sc{R_{1}}
       \end{array} $
       &
       $ \begin{array}{@{}l@{}}
               U_{\dollar} \rc{q_{1}}= \frac{1}{\sqrt{2}} \Sc{A_{1}} + \frac{1}{\sqrt{2}}\Sc{R_{1}} \\
               U_{\dollar} \rc{q_{2}}= \frac{1}{\sqrt{2}} \Sc{A_{2}} + \frac{1}{\sqrt{2}}\Sc{R_{2}} \\
               U_{\dollar} \rc{q_{3}}= \frac{1}{\sqrt{2}} \Sc{A_{3}} + \frac{1}{\sqrt{2}}\Sc{R_{3}}
       \end{array}$
       \\
\hline
       $ \begin{array}{@{}c@{}} \mbox{I} \\ ( \mathsf{path_{2}} ) \end{array} $
       &
       $ \begin{array}{@{}l@{}}
               U_{a} \rc{p_{1}} = \Sc{w_{1}} \\
               U_{a} \Sc{w_{1}} = \frac{1}{\sqrt{2}}\rc{p_{2}} + \frac{1}{2}
\Sc{A_{2}} + \frac{1}{2} \Sc{R_{2}} \\
               U_{a} \rc{p_{2}} = \Sc{w_{2}} \\
               U_{a} \Sc{w_{2}} = \frac{1}{\sqrt{2}}\rc{p_{2}} - \frac{1}{2}
\Sc{A_{2}} - \frac{1}{2} \Sc{R_{2}}
       \end{array} $
       &
       $ \begin{array}{@{}l@{}}
               U_{\dollar} \rc{p_{1}}= \frac{1}{\sqrt{2}} \Sc{A_{4}} + \frac{1}{\sqrt{2}}\Sc{R_{4}} \\
               U_{\dollar} \rc{p_{2}}= \frac{1}{\sqrt{2}} \Sc{A_{5}} + \frac{1}{\sqrt{2}}\Sc{R_{5}} \\
               U_{\dollar} \rc{p_{3}}= \frac{1}{\sqrt{2}} \Sc{A_{6}} + \frac{1}{\sqrt{2}}\Sc{R_{6}}
       \end{array}$
       \\
\hline
       $ \begin{array}{@{}c@{}} \mbox{II} \\ ( \mathsf{path_{1}} ) \end{array} $
       &
       $ \begin{array}{@{}l@{}}
               U_{a} \rc{q_{3}}=\Sc{w_{3}} \\
               U_{a} \Sc{w_{3}} = \frac{1}{\sqrt{2}}\rc{q_{4}} + \frac{1}{2}
\Sc{A_{3}} + \frac{1}{2} \Sc{R_{3}} \\
               U_{a} \rc{q_{4}}=\Sc{w_{4}} \\
               U_{a} \Sc{w_{4}} = \frac{1}{\sqrt{2}}\rc{q_{4}} - \frac{1}{2}
\Sc{A_{3}} - \frac{1}{2} \Sc{R_{3}}
       \end{array} $
       &
       $ \begin{array}{@{}l@{}}
               U_{\dollar} \rc{q_{4}} = \frac{1}{\sqrt{2}} \Sc{A_{7}} + \frac{1}{\sqrt{2}}\Sc{R_{7}} \\
               U_{\dollar} \rc{q_{5}}=  \frac{1}{\sqrt{2}} \Sc{A_{8}} + \frac{1}{\sqrt{2}}\Sc{R_{8}}
       \end{array} $
       \\
\hline
       $ \begin{array}{@{}c@{}} \mbox{II} \\ ( \mathsf{path_{2}} ) \end{array} $
       &
       $ \begin{array}{@{}l@{}}
               U_{a} \rc{p_{3}}=\frac{1}{\sqrt{2}} \rc{p_{4}} + \frac{1}{2}
\Sc{A_{4}} + \frac{1}{2} \Sc{R_{4}} \\
               U_{a} \rc{p_{4}}=\frac{1}{\sqrt{2}} \rc{p_{4}} - \frac{1}{2}
\Sc{A_{4}} - \frac{1}{2} \Sc{R_{4}}
       \end{array} $
       &
       $ \begin{array}{@{}l@{}}
               U_{\dollar} \rc{p_{4}} = \frac{1}{\sqrt{2}} \Sc{A_{9}} + \frac{1}{\sqrt{2}}\Sc{R_{9}} \\
               U_{\dollar} \rc{p_{5}}= \frac{1}{\sqrt{2}} \Sc{A_{10}} + \frac{1}{\sqrt{2}}\Sc{R_{10}}
       \end{array} $
       \\
\hline
       $ \begin{array}{@{}c@{}} \mbox{III} \\ ( \mathsf{path_{1}} ) \end{array} $
       &
       $ \begin{array}{@{}l@{}}
               U_{a} \rc{q_{5}}=\Sc{w_{5}} \\
               U_{a} \Sc{w_{5}} = \frac{1}{\sqrt{2}}\rc{q_{6}} + \frac{1}{2}
\Sc{A_{5}} + \frac{1}{2} \Sc{R_{5}} \\
               U_{a} \rc{q_{6}}=\Sc{w_{6}} \\
               U_{a} \Sc{w_{6}} = \frac{1}{\sqrt{2}}\rc{q_{6}} - \frac{1}{2}
\Sc{A_{5}} - \frac{1}{2} \Sc{R_{5}}
       \end{array} $
       &
       $ \begin{array}{@{}l@{}}
               U_{\dollar} \rc{q_{6}} = \frac{1}{\sqrt{2}} \Sc{A_{11}} + \frac{1}{\sqrt{2}}\Sc{R_{11}}
       \end{array} $
       \\
\hline
       $ \begin{array}{@{}c@{}} \mbox{III} \\ ( \mathsf{path_{2}} ) \end{array} $
       &
       $ \begin{array}{@{}l@{}}
               U_{a} \rc{p_{5}}=\frac{1}{\sqrt{2}} \rc{p_{6}} + \frac{1}{2}
\Sc{A_{6}} + \frac{1}{2} \Sc{R_{6}} \\
               U_{a} \rc{p_{6}}=\frac{1}{\sqrt{2}} \rc{p_{6}} - \frac{1}{2}
\Sc{A_{6}} - \frac{1}{2} \Sc{R_{6}}
       \end{array} $
       &
       $ \begin{array}{@{}l@{}}
               U_{\dollar} \rc{p_{6}} = \frac{1}{\sqrt{2}} \Sc{A_{12}} + \frac{1}{\sqrt{2}}\Sc{R_{12}}
       \end{array} $
       \\
\hline
       $ \begin{array}{@{}c@{}} \mbox{III} \\ ( \mathsf{path_{accept}} )
\end{array} $
       &
       $ \begin{array}{@{}l@{}}
               U_{a} \rc{a_{1}}=\frac{1}{\sqrt{2}} \rc{a_{2}} + \frac{1}{2}
\Sc{A_{7}} + \frac{1}{2} \Sc{R_{7}} \\
               U_{a} \rc{a_{2}}=\frac{1}{\sqrt{2}} \rc{a_{2}} - \frac{1}{2}
\Sc{A_{7}} - \frac{1}{2} \Sc{R_{7}} \\
               U_{a} \rc{a_{3}}=\frac{1}{\sqrt{2}} \rc{a_{4}} + \frac{1}{2}
\Sc{A_{8}} + \frac{1}{2} \Sc{R_{8}} \\
               U_{a} \rc{a_{4}}=\frac{1}{\sqrt{2}} \rc{a_{4}} - \frac{1}{2}
\Sc{A_{8}} - \frac{1}{2} \Sc{R_{8}}
       \end{array} $
       &
       $ \begin{array}{@{}l@{}}
               U_{\dollar} \rc{a_{1}} = \Sc{A_{17}} \\
               U_{\dollar} \rc{a_{3}} = \Sc{A_{18}} \\
               U_{\dollar} \rc{a_{2}} = \frac{1}{\sqrt{2}} \Sc{A_{13}} + \frac{1}{\sqrt{2}}\Sc{R_{13}} \\
               U_{\dollar} \rc{a_{4}} = \frac{1}{\sqrt{2}} \Sc{A_{14}} + \frac{1}{\sqrt{2}}\Sc{R_{14}}
       \end{array} $
       \\
\hline
       $ \begin{array}{@{}c@{}} \mbox{III} \\ ( \mathsf{path_{reject}} )
\end{array} $
       &
       $ \begin{array}{@{}l@{}}
               U_{a} \rc{r_{1}}=\frac{1}{\sqrt{2}} \rc{r_{2}} + \frac{1}{2}
\Sc{A_{9}} + \frac{1}{2} \Sc{R_{9}} \\
               U_{a} \rc{r_{2}}=\frac{1}{\sqrt{2}} \rc{r_{2}} - \frac{1}{2}
\Sc{A_{9}} - \frac{1}{2} \Sc{R_{9}} \\
               U_{a} \rc{r_{3}}=\frac{1}{\sqrt{2}} \rc{r_{4}} + \frac{1}{2}
\Sc{A_{10}} + \frac{1}{2} \Sc{R_{10}} \\
               U_{a} \rc{r_{4}}=\frac{1}{\sqrt{2}} \rc{r_{4}} - \frac{1}{2}
\Sc{A_{10}} - \frac{1}{2} \Sc{R_{10}}
       \end{array} $
       &
       $ \begin{array}{@{}l@{}}
               U_{\dollar} \rc{r_{1}} = \Sc{R_{17}} \\
               U_{\dollar} \rc{r_{3}} = \Sc{R_{18}} \\
               U_{\dollar} \rc{r_{2}} = \frac{1}{\sqrt{2}} \Sc{A_{15}} + \frac{1}{\sqrt{2}}\Sc{R_{15}} \\
               U_{\dollar} \rc{r_{4}} = \frac{1}{\sqrt{2}} \Sc{A_{16}} + \frac{1}{\sqrt{2}}\Sc{R_{16}}
       \end{array} $
       \\
       \hline
\end{tabular}
}
\caption{Specification of the transition function of the 1KWQFA for $ L_{NH} $ (part 1)}
\label{figure:1KWQFA-1}
\end{figure}
\begin{figure}[!h]       
       \setlength{\extrarowheight}{1pt}
       \centering
\footnotesize{
\begin{tabular}{|c|l|l|}
       \hline
       Stages & \multicolumn{2}{c|}{$ U_{b} $} \\
       \hline
       $ \begin{array}{@{}c@{}} \mbox{I} \\ ( \mathsf{path_{1}} ) \end{array} $
       &
       \multicolumn{2}{l|}{
               $ \begin{array}{@{}l@{}}
                       U_{b} \rc{q_{1}} = \frac{1}{\sqrt{2}} \Sc{A_{1}} + \frac{1}{\sqrt{2}}\Sc{R_{1}} \\
                       U_{b} \rc{q_{2}} = \rc{q_{3}} \\
                       U_{b} \rc{q_{3}} = \frac{1}{\sqrt{2}} \Sc{A_{2}} + \frac{1}{\sqrt{2}}\Sc{R_{2}}
               \end{array} $ }
       \\
       \hline
       $ \begin{array}{@{}c@{}} \mbox{I} \\ ( \mathsf{path_{2}} ) \end{array} $
       &
       \multicolumn{2}{l|}{
               $ \begin{array}{@{}l@{}}
                       U_{b} \rc{p_{1}} = \frac{1}{\sqrt{2}} \Sc{A_{3}} + \frac{1}{\sqrt{2}}\Sc{R_{3}} \\
                       U_{b} \rc{p_{2}} = \rc{p_{3}} \\
                       U_{b} \rc{p_{3}} = \frac{1}{\sqrt{2}} \Sc{A_{4}} + \frac{1}{\sqrt{2}}\Sc{R_{4}}
               \end{array} $
       }
       \\
       \hline
       $ \begin{array}{@{}c@{}} \mbox{II} \\ ( \mathsf{path_{1}} ) \end{array} $
       &
       \multicolumn{2}{l|}{
               $ \begin{array}{@{}l@{}}
                       U_{b} \rc{q_{4}} = \frac{1}{2}
\rc{q_{5}}+\frac{1}{2\sqrt{2}}\rc{a_{1}}+\frac{1}{2\sqrt{2}}\rc{r_{1}}
                               + \frac{1}{2} \Sc{A_{11}} + \frac{1}{2} \Sc{R_{11}}\\
                       U_{b} \rc{q_{5}} = \frac{1}{\sqrt{2}} \Sc{A_{5}} + \frac{1}{\sqrt{2}}\Sc{R_{5}}
               \end{array} $
       }
       \\
       \hline
       $ \begin{array}{@{}c@{}} \mbox{II} \\ ( \mathsf{path_{2}} ) \end{array} $
       &
       \multicolumn{2}{l|}{
               $ \begin{array}{@{}l@{}}
                       U_{b} \rc{p_{4}} = \frac{1}{2} \rc{p_{5}} + \frac{1}{2\sqrt{2}}\rc{a_{1}}
                       - \frac{1}{2\sqrt{2}}\rc{r_{1}} + \frac{1}{2} \Sc{A_{12}} +
\frac{1}{2} \Sc{R_{12}}\\
                       U_{b} \rc{p_{5}} = \frac{1}{\sqrt{2}} \Sc{A_{6}} + \frac{1}{\sqrt{2}}\Sc{R_{6}}
               \end{array} $
       }
       \\
       \hline
       $ \begin{array}{@{}c@{}} \mbox{III} \\ ( \mathsf{path_{1}} ) \end{array} $
       &
       \multicolumn{2}{l|}{
               $ \begin{array}{@{}l@{}}
                       U_{b} \rc{q_{6}} = \frac{1}{2}
\rc{q_{5}}+\frac{1}{2\sqrt{2}}\rc{a_{1}}+\frac{1}{2\sqrt{2}}\rc{r_{1}}
                               - \frac{1}{2} \Sc{A_{11}} - \frac{1}{2} \Sc{R_{11}}
               \end{array} $
       }
       \\
       \hline
       $ \begin{array}{@{}c@{}} \mbox{III} \\ ( \mathsf{path_{2}} ) \end{array} $
       &
       \multicolumn{2}{l|}{
               $ \begin{array}{@{}l@{}}
                       U_{b} \rc{p_{6}} = \frac{1}{2} \rc{p_{5}} + \frac{1}{2\sqrt{2}}\rc{a_{1}}
                               - \frac{1}{2\sqrt{2}}\rc{r_{1}} - \frac{1}{2} \Sc{A_{12}} -
\frac{1}{2} \Sc{R_{12}}
               \end{array} $
       }
       \\
       \hline
       $ \begin{array}{@{}c@{}} \mbox{III} \\ ( \mathsf{path_{accept}} )
\end{array} $
       &
       \multicolumn{2}{l|}{
               $ \begin{array}{@{}l@{}}
                       U_{b} \rc{a_{2}}=\frac{1}{\sqrt{2}} \rc{a_{3}} + \frac{1}{2}
\Sc{A_{13}} + \frac{1}{2} \Sc{R_{13}} \\
                       U_{b} \rc{a_{1}} = \frac{1}{\sqrt{2}} \Sc{A_{7}} + \frac{1}{\sqrt{2}}\Sc{R_{7}} \\
                       U_{b} \rc{a_{4}}=\frac{1}{\sqrt{2}} \rc{a_{3}} - \frac{1}{2}
\Sc{A_{13}} - \frac{1}{2} \Sc{R_{13}} \\
                       U_{b} \rc{a_{3}} = \frac{1}{\sqrt{2}} \Sc{A_{8}} + \frac{1}{\sqrt{2}}\Sc{R_{8}}
               \end{array} $
       }
       \\
       \hline
       $ \begin{array}{@{}c@{}} \mbox{III} \\ ( \mathsf{path_{reject}} ) \end{array} $
       &
       \multicolumn{2}{l|}{
               $ \begin{array}{@{}l@{}}
                       U_{b} \rc{r_{2}}=\frac{1}{\sqrt{2}} \rc{r_{3}} + \frac{1}{2}
\Sc{A_{14}} + \frac{1}{2} \Sc{R_{14}} \\
                       U_{b} \rc{r_{1}} = \frac{1}{\sqrt{2}} \Sc{A_{9}} + \frac{1}{\sqrt{2}}\Sc{R_{9}} \\
                       U_{b} \rc{r_{4}}=\frac{1}{\sqrt{2}} \rc{r_{3}} - \frac{1}{2}
\Sc{A_{14}} - \frac{1}{2} \Sc{R_{14}} \\
                       U_{b} \rc{r_{3}} = \frac{1}{\sqrt{2}} \Sc{A_{10}} + \frac{1}{\sqrt{2}}\Sc{R_{10}}
               \end{array} $
       }
       \\
       \hline
\end{tabular}
}
\caption{Specification of the transition function of the 1KWQFA for $ L_{NH} $ (part 2)}
\label{figure:1KWQFA-2}
\end{figure}

       Machine $ \mathcal{M} $ starts computation on symbol $ \cent $ by
branching into two paths,
       $ \mathsf{path_{1}} $ and $ \mathsf{path_{2}} $, with equal probability.
       Each path and their subpaths, to be described later, check whether
the input is of the form
       $ (aa^{*}b)(aa^{*}b)(aa^{*}b)^{*} $.
       The different stages of the program indicated in Figures
\ref{figure:1KWQFA-1} and \ref{figure:1KWQFA-2}
       correspond to the
       subtasks of this regular expression check. Stage I ends successfully
if the input begins with $ (aa^{*}b) $.
       Stage II checks the second $ (aa^{*}b)$. Finally, Stage III controls
whether the input ends with $ (aa^{*}b)^{*} $.

       The reader will note that many transitions in the machine are of the form
       \[ U_{\sigma} \ket{q_{i}} = \ket{\psi} + \alpha \ket{A_{k}} + \alpha
\ket{R_{k}}, \]
       where $ \ket{\psi} $ is a superposition of configurations such that $
\braket{\psi}{\psi}=1-2\alpha^{2} $,
       $ A_{k} \in Q_{a} $, $ R_{k} \in Q_{r} $.
       The equal-probability transitions to the ``twin halting states" $
A_{k} $ and $ R_{k} $ are included to ensure that
       the matrices are unitary, without upsetting the ``accept/reject
balance" until a final decision about
       the membership of the input in $ L_{NH} $ is reached. If the regular
expression check mentioned
       above fails, each path in question splits equiprobably to one
rejecting and one accepting configuration, and the overall probability
of
       acceptance of the machine
       turns out to be precisely $ \frac{1}{2} $.      If the input is indeed of
the form $ (aa^{*}b)(aa^{*}b)(aa^{*}b)^{*} $, whether
       the acceptance probability will exceed $ \frac{1}{2} $ or not
       depends on the following additional tasks performed by the
computation paths in order to test for the equality
       mentioned in the definition of $ L_{NH} $:
       \begin{enumerate}
               \item $ \mathsf{path_{1}} $ walks over the $ a $'s at the speed of
one tape square per step until
               reading the first $ b $. After that point, $ \mathsf{path_{1}} $
pauses for one step over each $ a $
               before moving on to the next symbol.
               \item $ \mathsf{path_{2}} $ pauses for one step over each $ a $
until reading the first $ b $.
               After that point, $ \mathsf{path_{2}} $ walks over each $ a $ at the
speed of one square per step.
               \item On each $ b $ except the first one, $ \mathsf{path_{1}} $ and
$ \mathsf{path_{2}} $ split to take the
               following two courses of action with equal probability:
               \begin{enumerate}
                       \item In the first alternative, $ \mathsf{path_{1}} $ and $
\mathsf{path_{2}} $
                       perform a two-way quantum Fourier transform (QFT) \cite{KW97}:
                       \begin{enumerate}
                               \item The targets of the QFT are two new computational paths,
i.e., $ \mathsf{path_{accept}} $
                               and $ \mathsf{path_{reject}} $. Disregarding the equal-probability
transitions to
                               the twin halting states mentioned above, the QFT is realized as:
                               \[ \mathsf{path_{1}} \rightarrow \frac{1}{\sqrt{2}}
\mathsf{path_{accept}} + \frac{1}{\sqrt{2}}
                                       \mathsf{path_{reject}}
                                \]
                               \[ \mathsf{path_{2}} \rightarrow \frac{1}{\sqrt{2}}
\mathsf{path_{accept}} - \frac{1}{\sqrt{2}}
                                       \mathsf{path_{reject}}
                                \]
                               \item $ \mathsf{path_{accept}} $ and $ \mathsf{path_{reject}} $
continue computation at the
                               speed of  $ \mathsf{path_{2}} $, walking over the $b$'s without
performing the QFT any more.
                       \end{enumerate}
                       \item In the second alternative, $ \mathsf{path_{1}} $ and $
\mathsf{path_{2}} $
                       continue computation without performing the QFT.
               \end{enumerate}
               \item On symbol $ \dollar $,
                       $ \mathsf{path_{accept}} $ enters an accepting state,
                       $ \mathsf{path_{reject}} $ enters a rejecting state,
                       $ \mathsf{path_{1}} $ and $ \mathsf{path_{2}} $ enter accepting and
rejecting states with
                       equal probability.
       \end{enumerate}

       Suppose that the input is of the form
       \[ w=a^{x}ba^{y_{1}}ba^{y_{2}}b \cdots a^{y_{t}}b, \]
       where $ x,t,y_{1}, \cdots, y_{t}  \in \mathbb{Z}^{+} $.

       $ \mathsf{path_{1}} $ reaches the first $ b $ earlier than $
\mathsf{path_{2}} $.
       Once it has passed the first $ b $, $ \mathsf{path_{2}} $ becomes
faster, and may or may not catch up with
       $ \mathsf{path_{1}} $, depending on the number of $ a $'s in the
input after the first $ b $.
       The two paths can meet on the symbol following the $ x^{th} $ $a$ after
the first $ b $, since at that point
       $ \mathsf{path_{1}} $ will have paused for the same number of steps
as $ \mathsf{path_{2}} $.
       Only if that symbol is a $ b $, the two paths will perform a QFT in
the same place and at the same time.
       To paraphrase, if there exists a $ k $ $ (1 \le k \le t) $ such that
$ x=\sum_{i=1}^{k}y_{i}\ $,
       $ \mathsf{path_{1}} $ and $ \mathsf{path_{2}} $ meet over the $
(k+1)^{th} $ $ b $, and perform the QFT
       at the same step. If there is no such $ k $, the paths either never
meet, or meet over an $ a $ without a QFT.

       The $ \mathsf{path_{accept}} $ and $ \mathsf{path_{reject}} $s that
are offshoots of $ \mathsf{path_{1}} $
       continue their traversal of the string faster than $
\mathsf{path_{1}} $. On the other hand,
       the offshoots of $ \mathsf{path_{2}} $ continue their traversal at
the same speed as $ \mathsf{path_{2}} $.

       By definition, the twin halting states reached during the computation
contribute equal amounts to the
       acceptance and rejection probabilities. $ \mathsf{path_{1}} $ and $
\mathsf{path_{2}} $ accept and reject
       equiprobably when they reach the end of the string. If  $
\mathsf{path_{1}} $ and $ \mathsf{path_{2}} $ never
       perform the QFT at the same time and in the same position, every QFT
produces two equal-probability paths which
       perform identical tasks, except that one accepts and the other one
rejects at the end.

       The overall acceptance and rejection probabilities are equal, $
\frac{1}{2} $, unless a $ \mathsf{path_{reject}} $
       with positive amplitude and a $ \mathsf{path_{reject}} $ with
negative amplitude can meet and therefore cancel
       each other. In such a case, the surviving $ \mathsf{path_{accept}}
$'s will contribute the additional acceptance
       probability that will tip the balance. As described above, such a
cancellation is only possible when
       $ \mathsf{path_{1}} $ and $ \mathsf{path_{2}} $ perform the QFT together.

       Therefore, if $ w \in L_{NH} $,
       the overall acceptance probability is greater than $ \frac{1}{2} $.
If $ w \notin L_{NH} $,
       the overall acceptance probability equals $ \frac{1}{2} $.
\qed\end{proof}

\begin{corollary}
	For any space bound $s$ satisfying $ s(n)=o(\log \log n) $, 
	\[ PrSPACE(s) \subsetneq PrQSPACE(s). \]
\end{corollary}

\begin{corollary}
	For any space bound $s$ satisfying $ s(n)=o(\log n) $,
	the class of languages recognized with unbounded error by 1PTMs is
	a proper subclass of the class of languages recognized with unbounded error by 1QTMs.
\end{corollary}

In the next section, we will prove a fact which will allow us to state
a similar inclusion relationship between the classes of languages
recognized by QTMs with restricted measurements and PTMs using
constant space.

\begin{theorem}
 The  language
	\[ L_{YS}=\{ a^{n-1}ba^{kn} \mid n>1,k > 0 \} \] is nonstochastic,
	and can be recognized by a 2KWQFA with unbounded error.
\end{theorem}
\begin{proof}
Suppose that $ L_{YS} $ is stochastic. Then, it is not hard to show that 
$ \{ a \} \mspace{-4mu} \cdot \mspace{-4mu} L_{YS} $ is stochastic, too.
However, as stated on page 88 of \cite{SS78}, $ \{ a \} \mspace{-4mu} \cdot \mspace{-4mu} L_{YS} $ 
is nonstochastic. 

We construct a 2KWQFA  $ \mathcal{M} = (Q,\Sigma,\delta,q_{0},Q_{a},Q_{r}) $,
where $ \Sigma=\{a,b\} $, and the state sets are
\begin{equation*}
	\begin{array}{l}
		Q_{n}  =  \{q_{0},q_{1},w_{1},w_{2},p_{1},p_{2},r_{1},r_{2},r_{3}\}, \\
		Q_{a}  =  \{A_{i} \mid 1 \le i \le 5 \},
		~Q_{r}=\{R_{i} \mid 1 \le i \le 5 \}.
	\end{array}
\end{equation*}
Let each $ U_{\sigma} $ induced by $ \delta $ 
act as indicated in Figure \ref{figure:L-YS}, and extend each to be unitary.

\begin{figure}[h]       
       \setlength{\extrarowheight}{1pt}
       \centering
\footnotesize{
\begin{tabular}{|c|l|l|}
\hline
       Stages & \multicolumn{1}{c|}{$ U_{\cent}, U_{a} $} &
                \multicolumn{1}{c|}{$ U_{b}, U_{\dollar} $} \\
\hline
	I & $ \begin{array}{@{}l@{}}
		U_{\cent} \ket{\rightstate{q_{0}}} = \ket{\rightstate{q_{0}}} \\
		U_{a} \ket{\rightstate{q_{0}}} = \frac{1}{\sqrt{2}} \ket{\rightstate{q_{1}}} + 
			\frac{1}{2} \Sc{A_{1}} + \frac{1}{2} \Sc{R_{1}} \\
		U_{a} \ket{\rightstate{q_{1}}} = \frac{1}{\sqrt{2}} \ket{\rightstate{q_{1}}} -
			\frac{1}{2} \Sc{A_{1}} - \frac{1}{2} \Sc{R_{1}} \\
	\end{array} $
	& $ \begin{array}{@{}l@{}}
		U_{b} \rc{q_{0}} = \frac{1}{\sqrt{2}} \Sc{A_{1}} + \frac{1}{\sqrt{2}} \Sc{R_{1}} \\
		U_{b} \rc{q_{1}} = \frac{1}{\sqrt{2}} \Sc{w_{1}} + \frac{1}{\sqrt{2}} \rc{r_{1}} \\
		U_{\dollar} \rc{q_{0}} = \frac{1}{\sqrt{2}} \Sc{A_{1}} + \frac{1}{\sqrt{2}} \Sc{R_{1}} \\
		U_{\dollar} \rc{q_{1}} = \frac{1}{\sqrt{2}} \Sc{A_{2}} + \frac{1}{\sqrt{2}} \Sc{R_{2}}
	\end{array} $
	\\ \hline
	& \multicolumn{2}{c|}{$ U_{\cent} $, $ U_{a} $, $ U_{b} $}
	\\ \hline
	$ \begin{array}{@{}c@{}} \mbox{II} \\ ( \mathsf{path_{left}} ) \end{array} $
	& \multicolumn{2}{l|}{ $ \begin{array}{@{}l@{}c@{}l@{}}
		U_{\cent} \lc{p_{1}} & = & \rc{p_{2}} \\
		U_{a} \lc{p_{1}} & = & \lc{p_{1}} \\
		U_{a} \rc{p_{2}} & = & \rc{p_{2}} \\
		U_{b} \Sc{w_{1}} & = & \Sc{w_{2}} \\
		U_{b} \Sc{w_{2}} & = & \frac{1}{\sqrt{2}} \lc{p_{1}} 
			- \frac{1}{2\sqrt{2}} \Sc{A_{2}} - \frac{1}{2\sqrt{2}} \Sc{R_{2}}
			- \frac{1}{2\sqrt{2}} \Sc{A_{3}} - \frac{1}{2\sqrt{2}} \Sc{R_{3}} \\
		U_{b} \rc{p_{2}} & = & \frac{1}{\sqrt{2}} \lc{p_{1}}
			+ \frac{1}{2\sqrt{2}} \Sc{A_{2}} + \frac{1}{2\sqrt{2}} \Sc{R_{2}}
			+ \frac{1}{2\sqrt{2}} \Sc{A_{3}} + \frac{1}{2\sqrt{2}} \Sc{R_{3}} \\
	\end{array} $ }
	\\ \hline
	& \multicolumn{2}{c|}{$ U_{a} $, $ U_{b} $, $ U_{\dollar} $}
	\\ \hline
	$ \begin{array}{@{}c@{}} \mbox{II} \\ ( \mathsf{path_{right}} ) \end{array} $
	& \multicolumn{2}{l|}{ $ \begin{array}{@{}l@{}c@{}l@{}}
		U_{a} \rc{r_{1}} & = & \frac{1}{\sqrt{2}} \rc{r_{2}} + \frac{1}{2} \Sc{A_{2}} + \frac{1}{2} \Sc{R_{2}} \\
		U_{a} \rc{r_{2}} & = & \frac{1}{\sqrt{2}} \rc{r_{2}} - \frac{1}{2} \Sc{A_{2}} - \frac{1}{2} \Sc{R_{2}} \\
		U_{a} \lc{r_{3}} & = & \lc{r_{3}} \\
		U_{b} \rc{r_{1}} & = & \frac{1}{\sqrt{2}} \Sc{A_{4}} + \frac{1}{\sqrt{2}} \Sc{R_{4}} \\
		U_{b} \rc{r_{2}} & = & \frac{1}{\sqrt{2}} \Sc{A_{5}} + \frac{1}{\sqrt{2}} \Sc{R_{5}} \\
		U_{b} \lc{r_{3}} & = & \frac{1}{\sqrt{2}} \Sc{A_{2}} - \frac{1}{\sqrt{2}} \Sc{R_{2}} \\
		U_{\dollar} \rc{r_{1}} & = & \frac{1}{\sqrt{2}} \Sc{A_{3}} + \frac{1}{\sqrt{2}} \Sc{R_{3}} \\
		U_{\dollar} \rc{r_{2}} & = & \lc{r_{3}} \\
	\end{array} $ }
	\\ \hline
\end{tabular}
}
\caption{Specification of the transition function of the 2KWQFA for $ L_{YS} $}
\label{figure:L-YS}
\end{figure}

If the input string does not begin with an $ a $, or if it contains  no $b$'s, the machine halts,
and the input is accepted with probability just $ \frac{1}{2} $.
Otherwise, the head moves to the right until it scans the first $ b $, on which the computation  splits to 
two equiprobable paths, say, $ \mathsf{path_{left}} $ and $ \mathsf{path_{right}} $.
Let the number of $ a $'s before the first $ b $ be $ n-1 > 0 $.

$ \mathsf{path_{left}} $ starts with two dummy stationary moves, and then enters an infinite loop.
In each iteration of this loop, the head goes to the left end-marker and then comes back to the $ b $ at the speed of 
one step per symbol. At the end of the $k^{th}$ iteration, exactly $ 2nk+n+3 $ steps after the start of computation, the head scans the $ b $ again, and $ \mathsf{path_{left}} $ splits to the superposition of configurations
\begin{equation*}
	\alpha_{k} \ket{p_{1},n} + \frac{\alpha_{k}}{2} \ket{A_{2},n+1}
	+ \frac{\alpha_{k}}{2} \ket{R_{2},n+1} + \frac{\alpha_{k}}{2} \ket{A_{3},n+1}
	+ \frac{\alpha_{k}}{2} \ket{R_{3},n+1},
\end{equation*}
where $ \alpha_{k} = \left( \frac{1}{\sqrt{2}} \right)^{n+k+1} $, and $\ket{s,h}$ denotes the configuration with state $s$ and head position $h$.

$ \mathsf{path_{right}} $ checks whether  the postfix of the input after the first $ b $ is of the form 
$ a^{+} $. If not, the machine halts, and the input is accepted with probability $ \frac{1}{2} $.
Otherwise, the head walks to the right end-marker and then comes back to the $ b $ at the speed of 
one step per symbol. 
Let the number of $ a $'s after the $ b $ be $ m > 0 $.
At the $ (2m+n+3)^{th} $ step, the head scans the $ b $, and $ \mathsf{path_{right}} $ 
splits to the superposition of configurations
\begin{equation*}
	\left( \frac{1}{\sqrt{2}} \right)^{m+n+1} \ket{A_{2},n+1} - 
	\left( \frac{1}{\sqrt{2}} \right)^{m+n+1} \ket{R_{2},n+1}.
\end{equation*}

The two paths can meet and interfere with each other only if 
\begin{equation*}
	2nk+n+3 = n+2m+3
\end{equation*}
or 
\begin{equation*}
	nk = m.
\end{equation*}
This is the case precisely for the members of $L_{YS} $, where the acceptance probability 
exceeds the rejection probability, similarly to what we had in the proof of Theorem \ref{theorem:1KWQFA}.
\qed\end{proof}

We do not know of a one-way QFA for $ L_{YS}$. Note that the somewhat simpler language  
$L_{fre}=\{ a^{n}ba^{n} \mid n \in \mathbb{Z}^{+} \}$ can be recognized with bounded error  by a 2PFA \cite{Fr81}.

The class C$_{=} $SPACE($s$) is defined \cite{Wa99}  as follows: A
language $L$ is in C$_{=} $SPACE($s$) if there exists a PTM that runs in
space $O(s)$, halts absolutely,\footnote{That is, for every input $w$,
there exists an integer $ k(w) $, such that the PTM halts with
probability 1 within $ k(w) $ steps.} and
accepts each input $w$ with probability precisely equal to $ \frac{1}{2} $ if and only if $ x \in L $.
We define the analogous family of quantum classes.

\begin{definition}
	\label{definition:c=qspace}
	A language $L$ is in C$_{=} $QSPACE($s$) if there exists a QTM that runs in space O($s$),
	halts absolutely, and accepts each input $w$ with probability precisely equal to $ \frac{1}{2} $ 
	if and only if $x \in L$.
\end{definition}

\begin{corollary}
	\label{corollary:coCSPACE-coCQSPACE}
	$ \mbox{coC}_{=}\mbox{SPACE}(1) \subsetneq \mbox{coC}_{=}\mbox{QSPACE}(1). $
\end{corollary}
\begin{proof}
	Since $ \mbox{coC}_{=}\mbox{SPACE}(1) $ is a proper subset of $ S $ \cite{Pa71}, 
	$ L_{NH} $ is not a member of $ \mbox{coC}_{=}\mbox{SPACE}(1) $.
	On the other hand, as shown in Theorem \ref{theorem:1KWQFA}, $ L_{NH} $
	is also a member of $ \mbox{coC}_{=}\mbox{QSPACE}(1) $.
\qed\end{proof}

\section{Languages recognized by RT-KWQFAs with unbounded error} \label{section:KWQFA-languages}

In this section, we settle an open problem of Brodsky and Pippenger
\cite{BP02}, giving a complete characterization of the class of
languages recognized with unbounded error by RT-KWQFAs. It turns out
that these restricted RT-QFAs, which are known to be inferior to RT-PFAs
in the bounded error case, are equivalent to them in the unbounded
error setting.

\begin{lemma}
       \label{lemma:rt-kwqfa}
       Any language recognized with cutpoint (or nonstrict cutpoint) 
       $ \frac{1}{2} $ by a RT-PFA with $ n $ internal states 
       can be recognized with cutpoint (or nonstrict cutpoint) $ \frac{1}{2} $ by a RT-KWQFA with $ O(n) $ internal states.
\end{lemma}
\begin{proof}
	Let $ L $ be a language recognized by an $ n $-state RT-PFA  
	\[ \mathcal{P} =(Q,\Sigma,\{A_{\sigma} \mid \sigma \in \tilde{\Sigma} \},q_{1},Q_{a}) \]
	with (nonstrict) cutpoint $ \frac{1}{2} $.
	We will construct a RT-KWQFA 
	\[ \mathcal{M}=(R, \Sigma, \{U_{\sigma} \mid \sigma \in \tilde{\Sigma} \}, r_{1},
		R_{a}, R_{r}) \] 
	which has $ 3n+6 $ internal states, and recognizes $ L $ with (nonstrict) cutpoint $ \frac{1}{2} $. The idea is to ``embed" the (not necessarily unitary) matrices $A_{\sigma}$ of the RT-PFA within the larger unitary matrices $U_{\sigma}$ of the RT-KWQFA.
	
	We define $ Q^{\prime} $, $ v_{0}^{\prime} $, and $ \{ A'_{\sigma} \mid \sigma \in \tilde{\Sigma} \} $ as follows:
	\begin{enumerate}
		\item $ Q^{\prime} = Q \cup \{ q_{n+1}, q_{n+2} \} $;
		\item $ v_{0}^{\prime} = (1,0,\ldots,0)^{T} $  is an $ (n+2) $-dimensional column vector;
		\item Each $ A_{\sigma}^{\prime} $ is a $ (n+2) \times (n+2) $-dimensional matrix:
			for each $ \sigma \in \Sigma \cup \{\cent\} $,
			\[			
             A'_{\sigma}= \left(
				\begin{array}{c|c}
					A_{\sigma} &  0_{n \times 2}  \\
					\hline
					0_{2 \times n} & I_{2 \times 2} \\
                     \end{array}
			\right)
			\]
			and 
			\[
             A_{\dollar}^{\prime}=\left(
				\begin{array}{c|c}
					0_{n \times n} & 0_{2 \times n}  \\
					\hline
					T_{2 \times n} & I_{2 \times 2} \\
				\end{array}
			\right)
			\left(
				\begin{array}{c|c}
					A_{\dollar} &  0_{n \times 2}  \\
					\hline
					0_{2 \times n} & I_{2 \times 2} \\
                     \end{array}
			\right),
			\]
     where $ T(1,i)=1 $ and $ T(2,i)=0 $ when $ q_{i} \in Q_{a} $, and
           $ T(1,i)=0 $ and $ T(2,i)=1 $ when $ q_{i} \notin Q_{a} $,
           for $ 1 \le i \le n $.
     \end{enumerate}
     For a given input $ w \in \Sigma^{*} $, 
     \begin{equation}
		v_{| \tilde{w} |}^{\prime} = A_{\dollar}^{\prime} A_{w_{|w|}}^{\prime} \cdots
			A_{w_{1}}^{\prime} A_{\cent}^{\prime} v_{0}'.
	\end{equation}
     It can easily be verified that
     \[ v_{| \tilde{w}|}^{\prime} = (0_{1 \times n} \mid f_{\mathcal{P}}(w), 1 - f_{\mathcal{P}}(w))^{T}. \]
     For each $ A^{\prime}_{\sigma} $, we will construct a 
     $ (n+2) \times (n+2) $-dimensional upper triangular matrix $ B_{\sigma} $
     so that the columns of 
     \begin{equation}
		\frac{1}{l} \left( \begin{array}{c} A_{\sigma}^{\prime} \\ \hline B_{\sigma}	\end{array} \right)
	\end{equation}
	form an orthonormal set, where $ l $ will be defined later.
	For this purpose, 
	the entries of  $ B_{\sigma} $, say $ b_{i,j} $ representing $ B_{\sigma}[i,j] $ for $ 1 \le i,j \le n+2 $,
	can be computed iteratively using the following procedure: 
     \begin{enumerate}
		\item Initialize all entries of $ B_{\sigma} $ to 0.
		\item Update the entries of $ B_{\sigma} $ to make the length of each column of 
			$ \left( \begin{array}{c} A_{\sigma}^{\prime} \\ \hline B_{\sigma} \end{array} \right) $ equal to $ l $
			and also to make the columns of 
			$ \left( \begin{array}{c} A_{\sigma}^{\prime} \\ \hline B_{\sigma} \end{array} \right) $
			pairwise orthogonal, by executing the following loop: 
			\\\\
			\begin{tabular}{ll}
				i.   & for $ i=1 $ to $ n+2 $ \\
				ii.  & ~~~~set $ l_{i} $ to the current length of the $ i^{th} $ column \\
				iii. & ~~~~set $ b_{i,i} $ to $ -\sqrt{l^{2}-l_{i}^{2}} $ \\
				iv.  & ~~~~for $ j=i+1 $ to $ n+2 $ \\
				v.   & ~~~~~~~~set $ b_{i,j} $ to some nonnegative value so that the $ i^{th} $ and $ j^{th} $ 
					 				\\
					 & ~~~~~~~~columns can become orthogonal \\
			\end{tabular} 
	\end{enumerate}
	The loop does not work properly if the value of $ l_{i} $, calculated at the (ii)$ ^{nd} $ step, 
	is greater than $ l $. Therefore, the value of $ l $ should be set carefully.
	For instance, by setting $ l $ to $ 2n+7 $, the following bounds can be easily verified
	for each iteration of the loop:
	\begin{itemize}
		\item $ l_{i} < 2 $ at the (ii)$ ^{nd} $ step;
		\item $ 2n+6 < |b_{i,i}| < 2n+7 $ at the (iii)$ ^{rd} $ step;
		\item $ 0 \le b_{j,i} < \frac{1}{n+3} $ at the (v)$ ^{th} $ step.
	\end{itemize}
	We define
	\begin{equation}
		U_{\sigma}=\left( 
			\begin{array}{c|c}				
				\begin{array}{c}
					A_{\sigma}^{\prime\prime}
					\\ \hline 
					B_{\sigma}^{\prime}
					\\ \hline
					B_{\sigma}^{\prime\prime}
				\end{array}
				 &
				D_{\sigma}
			\end{array}
		\right),
	\end{equation}
	where $ A_{\sigma}^{\prime\prime} = \frac{1}{l} A_{\sigma}^{\prime} $, 
	$ B_{\sigma}^{\prime} = B_{\sigma}^{\prime\prime} = \frac{1}{\sqrt{2}l} B_{\sigma} $, and 
	the entries of $ D_{\sigma} $ are selected to make $ U_{\sigma} $ 
	a unitary matrix.
	
	The state set $ R = R_{n} \cup R_{a} \cup R_{r}  $ is specified as:
	\begin{enumerate}
		\item $ r_{n+1} \in R_{a} $ corresponds to state $ q_{n+1} $;
		\item $ r_{n+2} \in R_{r} $ corresponds to state $ q_{n+2} $;
		\item $ \{r_{1},\ldots,r_{n}\} \in R_{n} $ correspond to the states of $ Q $, 
		where $ r_{1} $ is the start state;
		\item All the states defined for the rows of $ B_{\sigma}^{\prime} $ and $ B_{\sigma}^{\prime\prime} $ are
		respectively accepting and rejecting states.
	\end{enumerate}
	
	$ \mathcal{M} $ simulates the computation of $ \mathcal{P} $
	for the input string $ w $ by multiplying the amplitude of each non-halting state with
	$ \frac{1}{l} $ in each step. 
	Hence, the top $ n+2 $ entries of the state vector of $ \mathcal{M} $ equal
	\[	\left( \frac{1}{l} \right)^{|w|+2}
		\left(0_{1 \times n} \mid f_{\mathcal{P}}(w), 1- f_{\mathcal{P}}(w) \right)^{T}
	\]
	just before the last measurement on the right end-marker. 
	Note that, the halting states, except $ q_{n+1} $ and $ q_{n+2} $, will come in accept/reject pairs, so
	that transitions to them during the computation will add equal amounts to the overall 
	acceptance and rejection probabilities, and therefore will not affect the decision on the membership of the
	input in $ L $. We conclude that 
	\begin{equation}
		f_{\mathcal{M}}(w) > \frac{1}{2} \mbox{ if and only if } f_{\mathcal{P}}(w) > \frac{1}{2},
	\end{equation}
	and
	\begin{equation}
		f_{\mathcal{M}}(w) \geq \frac{1}{2} \mbox{ if and only if } f_{\mathcal{P}}(w) \geq \frac{1}{2}.
	\end{equation}
\qed\end{proof}

\begin{theorem}
       \label{theorem:SLUMM}
       The class of languages recognized by RT-KWQFAs with unbounded error is uS (uQAL).
\end{theorem}
\begin{proof}
       Follows from Lemma~\ref{lemma:rt-kwqfa}, Lemma \ref{lemma:RT-QFA-to-GFA} and \cite{Tu69}.
\qed\end{proof}

\begin{corollary}
	UMM = QAL $ \cap $ coQAL = S $ \cap $ coS.
\end{corollary}
\begin{proof}
	It is obvious that UMM $ \subseteq $ QAL $ \cap $ coQAL. Let $ L \in $ QAL $ \cap $ coQAL.
	Then, there exist two RT-KWQFAs $ \mathcal{M}_{1} $ and $ \mathcal{M}_{2} $ such that
	for all $ w \in L $, $ f_{\mathcal{M}_{1}} (w) > \frac{1}{2} $
	and $ f_{\mathcal{M}_{2}} (w) \geq \frac{1}{2} $, and
	for all $ w \notin L $, $ f_{\mathcal{M}_{1}} (w) \leq \frac{1}{2} $
	and $ f_{\mathcal{M}_{2}} (w) < \frac{1}{2} $.
	Let $ \mathcal{M}_{3} $ be a RT-KWQFA running $ \mathcal{M}_{1} $ and $ \mathcal{M}_{2} $
	with equal probability. Thus, we obtain that for all $ w \in L $,
	$ f_{\mathcal{M}_{3}} (w) > \frac{1}{2}$, and for all $ w \notin L $,
	$ f_{\mathcal{M}_{3}} (w) < \frac{1}{2} $. 
	Therefore, $ L \in  $ UMM.
\qed\end{proof}

Considering this result together with Theorem \ref{theorem:1KWQFA}, we
conclude that, unlike classical deterministic and probabilistic finite automata,
allowing the tape head to ``stay put" for
some steps during its left-to-right traversal of the input increases
the language recognition power of
quantum finite automata in the unbounded error case.

Since unbounded-error RT-PFAs and 2PFAs are equivalent in
computational power \cite{Ka91}, we are now able to state the
following corollary to Theorem \ref{theorem:1KWQFA}:
\begin{corollary}
	The class of languages recognized with unbounded error by
	constant-space PTMs is a proper subclass of the respective class for
	QTMs with restricted measurements.
\end{corollary}

Also note that, since the algorithm described in the proof of Theorem
\ref{theorem:1KWQFA} is
presented for a 1KWQFA, Corollary \ref{corollary:coCSPACE-coCQSPACE} is still valid when $
\mbox{coC}_{=}\mbox{QSPACE}(1) $
is defined for QTMs with restricted measurements.

\section{Concluding remarks} \label{section:ConcludingRemarks}

In this paper, we examined the capabilities of quantum Turing machines
operating under small space bounds in the
unbounded error setting. We proved that QTMs are strictly superior to
PTMs for all common space bounds that are $ o(\log \log n) $, and this
superiority extends to all sublogarithmic bounds when the machines are
allowed only one-way input head movement. We also gave a full
characterization of the class of languages recognized by real-time QFAs
employing restricted measurements; they turn out to be equivalent to
their probabilistic counterparts. It was also shown that
allowing the tape head to ``stay put" for some steps during its
left-to-right traversal of the input increases the
language recognition power of quantum finite automata in the unbounded
error case, allowing them to recognize some nonstochastic languages.
This means that two-way (and even one-way) QFAs are strictly more
powerful than RT-QFAs; whereas 2DFAs and unbounded-error 2PFAs are known
to be equivalent in power to their real-time versions \cite{Sh59,Ka91}.

While we have established some new results relating to the
relationship of probabilistic and quantum complexity classes in this
paper, the work reported here also gives rise to some new open
questions. As already mentioned, Watrous proved the equality
PrQSPACE($s$)=PrSPACE($s$) ($ s \in \Omega(\log n)$) for the cases where
PrQSPACE is defined in terms of Wa98-QTMs \cite{Wa98,Wa99}, and Wa03-QTMs \cite{Wa03}. 
We do not know how to prove these results for our
more general QTMs, and so the most that we can say about the
relationship among these classes now is PrSPACE($s$) $ \subsetneq $ PrQSPACE($s$) 
($ s \in o(\log \log n)$), and PrSPACE($s$) $ \subseteq $ PrQSPACE($s$) for all $s$. 
The only efficient simulation technique of a
quantum machine by a probabilistic machine that remains valid for our
definitions is that of Lemma \ref{lemma:RT-QFA-to-GFA}.

The reader may wonder why we did not present QTMs to be unidirectional by definition. 
The reason is that the known techniques \cite{BV97} for converting QTMs with arbitrary head movements to 
unidirectional QTMs do not work for the space-bounded case when the stationary 
``move" ($ \downarrow $) is included in the set of allowed head directions, 
and we stuck to the general definition to avoid any possibility of an unnecessary limitation 
of computational power.

After it was discovered in the context of this research, the
simulation method presented in Lemma \ref{lemma:rt-kwqfa} has been modified and
used in several contexts \cite{YS10A,YS10B,YS11B} to help establish
relationships between many different machine models.

Several real-time QFA variants have appeared in the literature. Our
results show that all of these which are at least as general as the
RT-KWQFA ($\mspace{-5mu}$\cite{Na99,Pa00,BMP03,Hi08}, and the real-time version
of the machines of \cite{AW02}) have the same computational power in
the unbounded error case. The class of languages recognized with
unbounded error by the weakest variant, the Moore-Crutchfield
QFA \cite{MC00}, is known \cite{BC01B} to be a proper subset of uS. 
One important model for
which no such characterization is yet known is the Latvian QFA \cite{ABGKMT06}.
Another question left open in this work
is the relationship between the computational powers of 1QFAs and 2QFAs.

\section*{Acknowledgements} \label{section:Acknowledgements}
We are grateful to Andris Ambainis and John Watrous for their helpful
comments on earlier versions of this paper.
We also thank R\={u}si\c{n}\v{s} Freivalds for his answers to our 
questions regarding nonstochastic languages, Maris Ozols for directing us to Watrous' lecture notes, Evgeniya Khusnitdinova for her kind help with Russian translations, and the two anonymous referees for their very useful suggestions.

\appendix

\section{Wellformedness conditions} \label{appendix:QTM-wellformedness}
\subsection{Local conditions for 2QFA wellformedness} \label{section:LC-2QFA-wellformedness}

Let $ c_{j_{1}} $ and $ c_{j_{2}} $ be two configurations, and
$ v_{j_{1}} $ and $ v_{j_{2}} $ be the corresponding columns of $ \mathsf{E} $ (See Figure \ref{figure:matrix-E}).
The value of $ v_{j_{1}}[i]  $ is determined by $ \delta $ if the $
i^{th} $ entry of $ v_{j_{1}} $
corresponds to a configuration to which $ c_{j_{1}} $ can evolve in one step,
and it is zero otherwise.
Let $ x_{1} $ and $ x_{2} $ be the positions of the input tape head
for the configurations $ c_{j_{1}} $ and $ c_{j_{2}} $, respectively.
In order to evolve to the same configuration in one step, the difference between
$ x_{1} $ and $ x_{2} $  must be at most 2.
Therefore, we obtain a total of three different cases, listed below, that completely
define the restrictions on the transition function.
Note that, by taking the conjugates of each summation,
we handle the symmetric cases that are shown in the parentheses.

For all $ q_{1},q_{2} \in Q; \sigma \in \tilde{\Sigma} $; 
(the summations are taken over $ q^{\prime} \in Q $; 
$ d \in \lrhd $; and $ \omega \in \Omega $),
\\
\footnotesize
1. $ x_{1} = x_{2} $:
\begin{equation}
	 \sum\limits_{q^{\prime} \in Q, d \in \lrhd,\omega \in \Omega}
      \overline{ \delta(q_{1},\sigma,q^{\prime},d,\omega) }
      \delta(q_{2},\sigma,q^{\prime},d,\omega) =
      \left\lbrace
              \begin{array}{ll}
                      1 & q_{1} = q_{2} \\
                      0 & \mbox{otherwise}
              \end{array}
      \right.
\end{equation}
\\
2. $ x_{1} = x_{2} - 1 $ ($ x_{1} = x_{2} + 1 $):
\begin{equation}
	\sum_{q^{\prime} \in Q,\omega \in \Omega} 
      \overline{\delta(q_{1},\sigma,q^{\prime},\rightarrow,\omega)} 
      \delta(q_{2},\sigma,q^{\prime},\downarrow,\omega) + 
      \overline{\delta(q_{1},\sigma,q^{\prime},\downarrow,\omega)} 
      \delta(q_{2},\sigma,q^{\prime},\leftarrow,\omega) = 0.
\end{equation}
\\
3. $ x_{1} = x_{2} - 2 $ ($ x_{1} = x_{2} + 2 $):
\begin{equation}
	\sum_{q^{\prime} \in Q,\omega \in \Omega}
	\overline{\delta(q_{1},\sigma,q^{\prime},\rightarrow,\omega)}
     \delta(q_{2},\sigma,q^{\prime},\leftarrow,\omega) = 0.
\end{equation}
\normalsize

\subsection{Unidirectional machines} \label{section:unidirectional-machines}

The wellformedness of unidirectional QTMs can be checked using the
simple conditions in Figure \ref{figure:uniQTM-wellformedness}.
Removing the reference to worktape symbols, we obtain the analogous
constraints for unidirectional 2QFAs as shown in Figure \ref{figure:uni2QFA-wellformedness}.

\begin{figure}[h]       
       \centering
       \fbox{
       \begin{minipage}{\textwidth}
               \footnotesize{
               For $ q_{1},q_{2} \in Q; \sigma \in \tilde{\Sigma};
\gamma_{1},\gamma_{2} \in \Gamma $,
               \begin{equation}
                       \sum\limits_{q^{\prime} \in Q,\gamma^{\prime} \in \Gamma,\omega \in \Omega}
                       \overline{ \delta(q_{1},\sigma,\gamma_{1},q^{\prime},\gamma^{\prime},\omega)
}
                       \delta(q_{2},\sigma,\gamma_{2},q^{\prime},\gamma^{\prime},\omega) =
                       \left\lbrace
                               \begin{array}{ll}
                                       1 & q_{1} = q_{2} \mbox{ and } \gamma_{1} = \gamma_{2} \\
                                       0 & \mbox{otherwise}
                               \end{array}
                       \right..
               \end{equation}
               }
       \end{minipage}
       }
       \caption{The local conditions for unidirectional QTM wellformedness}
       \label{figure:uniQTM-wellformedness}
\end{figure}

\begin{figure}[h]       
       \centering
       \fbox{
       \begin{minipage}{\textwidth}
               \footnotesize{
               For $ q_{1},q_{2} \in Q; \sigma \in \tilde{\Sigma} $,
               \begin{equation}
                       \sum\limits_{q^{\prime} \in Q,\omega \in \Omega}
                       \overline{ \delta(q_{1},\sigma,q^{\prime},\omega) }
                       \delta(q_{2},\sigma,q^{\prime},\omega) =
                       \left\lbrace
                               \begin{array}{ll}
                                       1 & q_{1} = q_{2} \\
                                       0 & \mbox{otherwise}
                               \end{array}
                       \right..
               \end{equation}
               }
       \end{minipage}
       }
       \caption{The local conditions for unidirectional 2QFA wellformedness}
       \label{figure:uni2QFA-wellformedness}
\end{figure}

As is the case with PTMs, the transition
function of a unidirectional QTM can be specified easily by transition matrices of the
form
$ \{ E_{\sigma,\omega} \} $, whose rows and columns are indexed by
(internal state, work tape symbol) pairs
for each $ \sigma \in \tilde{\Sigma} $ and $ \omega \in \Omega $. 
It can be verified that the wellformedness condition is then equivalent to the requirement that, 
for each  $ \sigma \in \tilde{\Sigma} $,
\begin{equation}
       \sum_{\omega \in \Omega} E_{\sigma,\omega}^{\dagger} E_{\sigma,\omega} = I.
\end{equation}

Similarly, for each $ \sigma \in \tilde{\Sigma} $ and $ \omega \in \Omega $,
well-formed unidirectional 2QFAs can be described by transition matrices of the form
$ \{ E_{\sigma,\omega} \} $, whose rows and columns are indexed by
internal states,
such that for each  $ \sigma \in \tilde{\Sigma} $,
\begin{equation}
       \sum_{\omega \in \Omega} E_{\sigma,\omega}^{\dagger} E_{\sigma,\omega} = I.
\end{equation}

\section{CQTMs} \label{appendix:CQTM}

To specialize our general QTM model in order to ensure that the head positions are classical,
we associate combinations of head movements with measurement outcomes.
There are 9 different pairs of possible movement directions ($
\lrhd^{2} = \{ \leftarrow, \downarrow, \rightarrow \} \times \{
\leftarrow, \downarrow, \rightarrow \} $) for the input and work tape
heads,
and so we can classify register symbols with the function
\begin{equation}
       D_{r} : \Omega \rightarrow {\lrhd}^{2}.
\end{equation}
We have $ D_{r}( \omega ) = ( \downarrow, \downarrow ) $
if $ \omega \in \Omega_{a} \cup \Omega_{r} $.
We split $ \Omega_{n} $ into $ 9 $ parts, i.e.
\begin{equation}
       \Omega_{n} = \bigcup\limits_{d_{i},d_{w} \in \lrhd }
\Omega_{n,d_{i},d_{w}},
\end{equation}
where
\begin{equation}
       \Omega_{n,d_{i},d_{w}} = \{ \omega \in \Omega_{n} \mid D_{r}(\omega)=(d_{i},d_{_{w}})  \}.
\end{equation}
Therefore, the outcome set will have 11 elements, represented as
triples, specified as follows:
\begin{enumerate}
       \item ``$ (n,d_{i},d_{w}) $": the computation continues and the
positions of the input and work tape heads
               are updated     with respect to $ d_{i} $ and $ d_{w} $, respectively;
       \item ``$ (a,\downarrow,\downarrow) $": the computation halts and the
input is accepted with no head movement;
       \item ``$ (r,\downarrow,\downarrow) $": the computation halts and the
input is rejected with no head movement.
\end{enumerate}

The transition function of CQTMs will be
specified so that
when the CQTM is in state $ q $ and reads $ \sigma $ and $ \gamma $
respectively on the input and work tapes,
it will enter state $ q^{\prime} $, and write $ \gamma^{\prime} $ and
$ \omega $ respectively
on the work tape and the finite register with the amplitude
\begin{equation}
       \delta(q,\sigma,\gamma,q^{\prime},\gamma^{\prime},\omega) \in \tilde{\mathbb{C}} .
\end{equation}
Since the update of the positions of the input and work tape heads is
performed classically, it is no longer a part of the transitions.
Note that the transition function of 2QFAs with classical head
(2CQFAs) \cite{AW02}
is obtained by removing the mention of the work tape from the above description.

Moreover, as with unidirectional QTMs (resp. unidirectional 2QFAs),
for each $ \sigma \in \tilde{\Sigma} $ and $ \omega \in \Omega $,
CQTMs (2CQFAs) can be described by transition matrices
$ \{ E_{\sigma,\omega} \} $ satisfying the same properties. (See Appendix \ref{appendix:QTM-wellformedness}.)

As also argued in \cite{Wa03},
CQTMs are sufficiently general for simulating any classical TM.
We will present a trivial simulation.

\begin{lemma}
       CQTMs can simulate any PTM exactly.
\end{lemma}
\begin{proof}
Let $ \mathcal{P} = (Q, \Sigma, \Gamma, \delta_{\mathcal{P}},q_{1}, Q_{a}, Q_{r}) $ be a PTM. We build a CQTM
$ \mathcal{M} = (Q, \Sigma, \Gamma, \Omega, \delta_{\mathcal{M}},q_{1},\Delta ) $.
For each $ (q,\gamma,q^{\prime},\gamma^{\prime}) \in Q \times \Gamma
\times Q \times \Gamma $,
we define a register symbol $ \omega_{(q,\gamma,q^{\prime},\gamma^{\prime})} $ such that
\begin{enumerate}
       \item if $ q^{\prime} \in Q_{a} $: $ \omega_{(q,\gamma,q^{\prime},\gamma^{\prime})} \in
\Omega_{(a,\downarrow,\downarrow)} $;
       \item if $ q^{\prime} \in Q_{r} $: $ \omega_{(q,\gamma,q^{\prime},\gamma^{\prime})} \in
\Omega_{(r,\downarrow,\downarrow)} $;
       \item if $ q^{\prime} \in Q_{n} $: $ \omega_{(q,\gamma,q^{\prime},\gamma^{\prime})} \in
\Omega_{(n,D_{i}(q^{\prime}), D_{w}(q^{\prime}))} $.
\end{enumerate}
We conclude with setting
\begin{equation}
       \delta_{\mathcal{M}}(q,\sigma,\gamma,q^{\prime},\gamma^{\prime},\omega_{(q,\gamma,q^{\prime},\gamma^{\prime})}) =
       \sqrt{ \delta_{\mathcal{P}}(q,\sigma,\gamma,q^{\prime},\gamma^{\prime})}
\end{equation}
for each $\sigma \in \Sigma$, and setting the values of $\delta_{\mathcal{M}}$ that are still undefined to zero.
\qed\end{proof}
This result is also valid for two-way and real-time finite automata:
\begin{corollary}
       2CQFAs (RT-QFAs) can simulate any 2PFA (RT-PFA) exactly.
\end{corollary}

\bibliographystyle{plain}
\bibliography{YakaryilmazSay}

\end{document}